\PassOptionsToPackage{dvipsnames}{xcolor}
\NewCommandCopy{\LaTeXtextbf}{\textbf}
\documentclass{ieeeaccess}
\usepackage[T1]{fontenc}
\usepackage{graphicx}
\usepackage{amsmath,amssymb,amsfonts}
\usepackage{textcomp}
\usepackage{tikz}
\NewSpotColorSpace{PANTONE}
  \AddSpotColor{PANTONE} {PANTONE3015C} {PANTONE\SpotSpace 3015\SpotSpace C} {1 0.3 0 0.2}
  \SetPageColorSpace{PANTONE}%
\usepackage{bm}
\makeatletter
\AtBeginDocument{\DeclareMathVersion{bold}
\SetSymbolFont{operators}{bold}{T1}{times}{b}{n}
\SetSymbolFont{NewLetters}{bold}{T1}{times}{b}{it}
\SetMathAlphabet{\mathrm}{bold}{T1}{times}{b}{n}
\SetMathAlphabet{\mathit}{bold}{T1}{times}{b}{it}
\SetMathAlphabet{\mathbf}{bold}{T1}{times}{b}{n}
\SetMathAlphabet{\mathtt}{bold}{OT1}{pcr}{b}{n}
\SetSymbolFont{symbols}{bold}{OMS}{cmsy}{b}{n}
\renewcommand\boldmath{\@nomath\boldmath\mathversion{bold}}}
\makeatother
\usepackage{orcidlink}
\usepackage[hyperfirst=false]{glossaries}
\usepackage{tablefootnote}
\usepackage{amsmath,amssymb,amsthm}
\usepackage{mathtools,mathrsfs,bm}
\usepackage{booktabs,multirow,threeparttable,tabularx}
\usepackage{algorithm,algpseudocode}
\usepackage{tikz}
\usetikzlibrary{shapes,arrows,positioning,fit,backgrounds,arrows.meta}
\usepackage{caption,subcaption,enumitem}
\usepackage{cite}
\usepackage{hyperref}
\usepackage[capitalise,noabbrev]{cleveref}
\usetikzlibrary{shapes,arrows,positioning,calc,arrows.meta,backgrounds}

\usepackage{tikz}

\newcommand{\F}{\mathcal{F}}
\newcommand{\I}{\mathcal{I}}
\newcommand{\argmin}{\operatorname{argmin}}
\newcommand{\argmax}{\operatorname{argmax}}
\newcommand{\GED}{\operatorname{GED}}
\newcommand{\wt}{W_{\text{total}}}
\newcommand{\degr}{\operatorname{deg}}

\usepackage{cleveref}
\newtheorem{theorem}{Theorem}
\newtheorem{lemma}{Lemma}
\newtheorem{corollary}{Corollary}
\newtheorem{proposition}{Proposition}
\newtheorem{definition}{Definition}
\theoremstyle{remark}

\newtheorem{heuristic}{Heuristic Rule}
\theoremstyle{plain}

\algnewcommand{\Break}{\textbf{break}}

\title{Graph Edit Distance Formulation for the Vehicle Routing Problem:
  Theory and Analysis}
\author{\uppercase{Adel Dabah}\authorrefmark{1}}

\address[1]{ \textit{J\"{u}lich Supercomputing Centre} \textit{Forschungszentrum J\"{u}lich}  J\"{u}lich, Germany\\ (e-mail: a.dabah@fz-juelich.de)}

\corresp{Corresponding author: Adel Dabah (e-mail:  a.dabah@fz-juelich.de).}

\begin{document}
\doi{10.1109/ACCESS.2026.xxx}
\begin{abstract}
We show that the Vehicle Routing Problem (VRP) can be reformulated as a Graph Edit Distance (GED) maximization problem. Under a simple edge‑deletion cost model, minimizing total route cost is equivalent to maximizing the total weight of edges deleted from the complete instance graph. This formulation models VRP at the edge level, where solutions are defined by selected edges rather than route sequences, enabling structural analyses that are difficult in classical formulations: per‑edge attribution of solution quality, decomposition of the optimality gap, characterization of solution sparsity, and identification of edges that are hard to reach by greedy construction.

Theoretically, we establish a merge‑decomposition theorem showing that Clarke‑Wright savings equal per‑merge GED increments, and an approximation‑transfer theorem that turns GED approximation ratios into VRP cost bounds. Using this reformulation, we analyze 90 CVRP benchmark instances with known optimal solutions. We find that optimal routing graphs use only 5.5\% of available edges, that approximately 3.0\% of optimal edges are consistently not found by Clarke‑Wright heuristics under repeated restarts, and that the cost gap decomposes into missed optimal edges and substituted non‑optimal edges of comparable total weight. The edge‑additive objective provides a natural per‑edge supervision signal for future graph neural network approaches to edge prediction, suggesting a  potential connection to graph neural network approaches that we leave for follow-up work.
\end{abstract}

\maketitle

\section{Introduction}
\IEEEPARstart{T}{he} Vehicle Routing Problem (VRP) is a central combinatorial optimization problem with applications in logistics, transportation, and service systems~\cite{toth2002vehicle}. Despite decades of progress and highly effective heuristics, VRP remains computationally challenging, with most variants classified as NP-hard~\cite{lenstra1981complexity}. The standard route-based formulation provides limited visibility into the edge-level structure of solutions, making it difficult to analyze how individual edges contribute to solution quality or how heuristic solutions deviate from optimal ones.

Concurrently, Graph Edit Distance (GED) has emerged as a powerful framework in pattern recognition and machine learning for measuring similarity/dissimilarity between graphs~\cite{bunke1999error}. GED is defined as the minimum cost set of edit operations (node/edge insertions, deletions, and substitutions) required to transform one graph into another. While GED itself is NP-hard~\cite{bunke2011recent}, it has efficient approximation algorithms and has found applications in chemical informatics, computer vision, and social network analysis.

Traditional graph-based VRP formulations~\cite{toth2002vehicle,laporte1992vehicle,vidal2013metaheuristics} model the instance as a complete weighted graph and solutions as sets of routes (cycles). In this work, we introduce an edge-centric perspective based on a GED formulation, by reformulating the VRP objective as a GED maximization problem. This reformulation shifts the focus from route construction to edge selection, enabling structural analyses of solution quality at the level of individual edges.  While similar edge-level analyses can in principle be derived from classical VRP solutions, the GED formulation makes them directly and naturally available through the additive structure of the objective, without requiring separate post-processing steps. 

The classical VRP objective $C(\mathcal{R}) = \sum_k \sum_i w(v_i, v_{i+1})$ is path-aggregated: each edge's contribution depends on its position in a route sequence. The GED objective $\sum_{e \in E \setminus E_t} w(e)$ is edge-additive: each edge contributes independently. This distinction enables per-edge attribution of solution quality, per-edge supervision for learning, and direct application of GED approximation bounds with known cost guarantees.

This paper provides the following contributions:
\begin{itemize}
\item A formal equivalence between minimizing VRP route cost and maximizing GED (\cref{thm:main}), including a cost-preserving correspondence between VRP solutions and feasible routing graphs (\cref{lem:correspondence}).
\item A merge-decomposition theorem showing Clarke--Wright savings are exactly per-merge GED increments (\cref{thm:merge}), and an approximation-transfer theorem providing VRP cost bounds from GED approximation ratios (\cref{thm:approximation}).
\item Using this reformulation, we conduct an edge-level structural analysis of optimal solutions on 90 benchmark instances, revealing sparsity (5.5\% of edges used), reachability limits (3.0\% not found by greedy construction under tested restarts), and a decomposition of the cost gap into missed vs.\ wasted edges.
\item A proposed path from GED to learned edge prediction via GNNs, where the edge-additive objective suggests a natural per-edge loss function and the reachability results identify supervision targets. No GNN experiments are included; this direction is left for future work
\end{itemize}

As a proof of concept, we use a standard Clarke--Wright multi-start heuristic solely to generate solutions for analysis. The goal is not to propose a new state-of-the-art solver, but to demonstrate how the proposed reformulation reveals structural properties of VRP solutions.

The remainder of this paper is organized as follows: \cref{sec:prelim} reviews background on VRP and GED; \cref{sec:formulation} presents the equivalence formulation; \cref{sec:proof} provides the formal theorem and proof; \cref{sec:experiments} presents experimental validation with structural analysis; \cref{sec:ged-perspective} develops the GED perspective on solution methods including edge frequency analysis and GNN connections; \cref{sec:related} reviews related work; and \cref{sec:conclusion} concludes the paper.
\section{Preliminaries}\label{sec:prelim}
In this section, we present basic definitions about both VRP and GED formulations.

\subsection{Vehicle Routing Problem (VRP)}\label{sec:vrp}

The Vehicle Routing Problem (VRP) was introduced by Dantzig and Ramser in 1959 to model the optimal delivery of goods from a central depot to geographically dispersed customers using a fleet of identical vehicles~\cite{dantzig1959truck}.  In its basic form, the problem asks for a set of vehicle routes, each starting and ending at the depot, such that every customer is visited exactly once, the total demand on any route does not exceed the vehicle capacity, and the sum of travel costs is minimized.

A standard VRP instance is defined by the tuple \(\I = (V,E,w,m,Q,D)\), where:
\begin{itemize}
    \item \(V = \{0,1,\ldots,n\}\) is the vertex set, with vertex \(0\) representing the depot and vertices \(\{1,\ldots,n\}\) representing customers.
    \item \(E = \{(i,j) \mid i,j \in V, i \neq j\}\) is the complete edge set. In this paper we use a \textbf{directed} graph model: edges are ordered pairs, but weights are symmetric \(w(i,j)=w(j,i)\).
    \item \(w: E \to \mathbb{R}^{+}\) is the non-negative cost (distance, time) function satisfying the triangle inequality: 
    \[w(i,j) \leq w(i,k) + w(k,j) \quad \forall i,j,k \in V\]
    \item \(m\in \mathbb{Z}^{+}\) is the number of identical vehicles available.
    \item \(Q\in \mathbb{R}^{+}\) is the vehicle capacity.
    \item \(D = \{d_1,\ldots,d_n\}\) are customer demands with \(d_i \leq Q\) for all \(i \in \{1,\ldots,n\}\).
\end{itemize}

The objective is to find a set of routes that satisfies all constraints and minimizes the total distance traveled. Despite its simple description, the VRP is NP-hard, and most of its variants remain computationally challenging even for moderate instance sizes.

\begin{definition}
\label{def:vrp-solution}
A feasible solution \(\mathcal{R} = \{R_1,\ldots,R_m\}\) consists of \(m\) routes, where each route \(R_k = (v_0^k, v_1^k, \ldots, v_{l_k}^k)\) satisfies:
\begin{enumerate}
    \item \(v_0^k = v_{l_k}^k = 0\) (start and end at depot).
    \item \(\{v_1^k,\ldots,v_{l_k-1}^k\} \subseteq \{1,\ldots,n\}\) (intermediate nodes are customers).
    \item \(\bigcup_{k=1}^m \{v_1^k,\ldots,v_{l_k-1}^k\} = \{1,\ldots,n\}\) (all customers served).
    \item \(\{v_1^k,\ldots,v_{l_k-1}^k\} \cap \{v_1^{k'},\ldots,v_{l_{k'}-1}^{k'}\} = \emptyset\) for \(k \neq k'\) (customers served exactly once).
    \item \(\sum_{i \in R_k} d_i \leq Q\) for all \(k\) (capacity constraint).
\end{enumerate}
The objective is to minimize total route cost:
\begin{equation}
\min_{\mathcal{R}} \sum_{k=1}^m \sum_{i=0}^{l_k-1} w(v_i^k, v_{i+1}^k) \label{eq:vrp-objective}
\end{equation}
\end{definition}

\subsection{Graph Edit Distance (GED)}\label{sec:ged}

Graph Edit Distance (GED) provides a flexible measure of similarity between graphs. A graph \(G = (V, E)\) consists of a set of vertices \(V\) and a set of edges \(E \subseteq V \times V\). In many applications, vertices and edges may carry labels or attributes; we denote by \(\ell(v)\) and \(\ell(e)\) the labels of a vertex \(v\) and an edge \(e\), respectively. These labels can represent, for example, identifiers, colors, or numerical attributes.

The edit operations that transform one graph into another are:
\begin{itemize}
    \item \textbf{Vertex deletion}: remove a vertex and all its incident edges.
    \item \textbf{Vertex insertion}: add a new vertex, optionally with a label.
    \item \textbf{Vertex substitution}: replace the label of a vertex with another label.
    \item \textbf{Edge deletion}, \textbf{edge insertion}, and \textbf{edge substitution}: analogous operations for edges.
\end{itemize}
Each operation has an associated non‑negative cost \(c(\cdot)\). An edit path is a sequence of such operations that transforms a source graph \(G_1\) into a target graph \(G_2\). The total cost of an edit path is the sum of the costs of its constituent operations.

The Graph Edit Distance between \(G_1\) and \(G_2\) is defined as the minimum cost over all edit paths that transform \(G_1\) into \(G_2\):

\begin{equation}
\GED(G_1, G_2) = \min_{e \in \Gamma(G_1, G_2)} \sum_{op \in e} c(op) \label{eq:ged}
\end{equation}
where \(\Gamma(G_1, G_2)\) is the set of all edit paths transforming \(G_1\) into \(G_2\).

When the vertex sets of the two graphs correspond one‑to‑one (as will be the case in our VRP formulation), vertex edit operations are unnecessary, and GED reduces to a problem on edges only. In that setting, for edge sets \(E_1\) and \(E_2\) we have the simplified expression
\begin{equation}
\GED(G_1, G_2) = \sum_{e \in E_1 \setminus E_2} c_{\text{del}}(e) + \sum_{e \in E_2 \setminus E_1} c_{\text{ins}}(e) \label{eq:ged-simple}
\end{equation}
where \(c_{\text{del}}(e)\) and \(c_{\text{ins}}(e)\) are the costs of deleting or inserting an edge \(e\), respectively. This form is the foundation for our VRP–GED equivalence.

\section{VRP-GED Equivalence Formulation}\label{sec:formulation}

We now introduce two graph structures that capture the VRP instance and any feasible solution. Both graphs share the same vertex set \(V = \{0,1,\dots,n\}\), where vertex \(0\) is the depot and the remaining vertices represent customers.

\paragraph{Complete VRP graph \(G_s\).}
The source graph \(G_s\) encodes all possible travel connections. It is defined as the complete weighted directed graph on the vertex set \(V\) with edge weights given by the distance function \(w\):
\[
G_s = (V,\; E,\; w),\qquad E = \{(i,j) \mid i,j\in V,\ i\neq j\}.
\]
Thus \(G_s\) contains every potential directed edge that could be used in a vehicle route. Edge weights are symmetric: \(w(i,j)=w(j,i)\).

\paragraph{Routing graph \(G_t\).}
Given a feasible VRP solution \(\mathcal{R} = \{R_1,\ldots,R_m\}\), the target graph \(G_t\) captures exactly the directed edges that are actually traversed. It is defined as the subgraph of \(G_s\) induced by the union of all route edges:
\[
G_t = (V,\; E_t,\; w_t),
\]
where:
\begin{itemize}
    \item \(E_t = \bigcup_{k=1}^m \bigl\{(v_i^k, v_{i+1}^k) \mid i = 0,\ldots, l_k-1\bigr\}\) (each directed edge appears once),
    \item \(w_t(e) = w(e)\) for every \(e \in E_t\) (the weights are inherited from \(G_s\)).
\end{itemize}
Because the routes form directed cycles that start and end at the depot, every customer vertex appears with in-degree = out-degree = 1 in \(G_t\); the depot has in-degree = out-degree = \(m\). Hence \(G_t\) is a proper subgraph of \(G_s\) (\(E_t \subseteq E\)) that satisfies specific degree and connectivity conditions.

The following characterization summarises these necessary and sufficient conditions for a subgraph to correspond to a feasible VRP solution.

\begin{definition}[Feasible Routing Graph Characterization]\label{def:feasible-characterization}
A directed subgraph \(G_t = (V, E_t)\) of the complete graph \(G_s\) is a feasible routing graph if and only if:
\begin{enumerate}
    \item \(\degr^+(0)=\degr^-(0) = m\) (depot in-degree = out-degree = number of vehicles);
    \item \(\degr^+(i)=\degr^-(i) = 1\) for all \(i \in \{1,\ldots,n\}\) (every customer has degree exactly 1 in each direction);
    \item \(G_t\) decomposes into \(m\) edge-disjoint directed cycles, each starting and ending at the depot (the cycles share only the depot vertex);
    \item For each cycle \(C_k\) with customer set \(V_{C_k}\setminus\{0\}\), the total demand \(\sum_{i \in V_{C_k} \setminus \{0\}} d_i \leq Q\) (capacity constraint).
\end{enumerate}
\end{definition}

Note that $\mathcal{F}$ (or $\F$) denotes the set of \emph{feasible routing graphs}, i.e., subgraphs satisfying Definition~\ref{def:feasible-characterization}. When referring to the set of feasible VRP solutions (route lists), we use $\mathcal{F}_{\mathrm{VRP}}$. The correspondence of Lemma~\ref{lem:correspondence} identifies these two sets up to route-order and reversal symmetries.

\begin{lemma}[Cost-Preserving Correspondence]\label{lem:correspondence}
The mapping $\Phi:\mathcal{R}\mapsto G_t$ that sends a feasible VRP
solution to its routing graph (\cref{def:feasible-characterization})
satisfies:
\begin{enumerate}[label=(\alph*),nosep]
\item \emph{Cost preservation:} $C(\mathcal{R})=\sum_{e\in E_t}w(e)$ for every feasible $\mathcal{R}$.
\item \emph{Surjectivity onto $\F$:} every $G_t\in\F$ is $\Phi(\mathcal{R})$ for some feasible $\mathcal{R}$.
\item \emph{Quotient injectivity:} $\Phi(\mathcal{R})=\Phi(\mathcal{R}')$ if and only if $\mathcal{R}$ and $\mathcal{R}'$ differ by (i) a permutation of the route list and/or (ii) reversal of one or more cycles.
\end{enumerate}
Consequently $\Phi$ induces a bijection between the quotient set
$\F_{\mathrm{VRP}}/{\sim}$ (solutions modulo route-order and reversal
symmetries) and $\F$, and this bijection is cost-preserving.
\end{lemma}

\begin{proof}
\textbf{(a) Cost preservation.}
By construction $E_t=\bigcup_k\{(v_i^k,v_{i+1}^k)\}$ as a multiset; in
the directed model these arcs are pairwise distinct (\cref{def:feasible-characterization}),
so the union is a true set. The route cost
$C(\mathcal{R})=\sum_k\sum_i w(v_i^k,v_{i+1}^k)$ then equals
$\sum_{e\in E_t}w(e)$ term-by-term.

\textbf{(b) Surjectivity.}
Let $G_t\in\F$. Conditions (1)--(3) of \cref{def:feasible-characterization}
give an Eulerian decomposition of $G_t$ into $m$ arc-disjoint directed
cycles $C_1,\ldots,C_m$ through the depot. Each $C_k$ written starting
at the depot, $C_k=(0,v_1^k,\ldots,v_{l_k-1}^k,0)$, is a route $R_k$;
condition (4) supplies the capacity constraint. Hence
$\mathcal{R}=\{R_1,\ldots,R_m\}$ is a feasible VRP solution with
$\Phi(\mathcal{R})=G_t$.

\textbf{(c) Quotient injectivity.}
If $\mathcal{R}'$ is obtained from $\mathcal{R}$ by permuting the route
list, the union $\bigcup_k E_{R_k}$ is unchanged, so
$\Phi(\mathcal{R}')=\Phi(\mathcal{R})$. If $\mathcal{R}'$ is obtained
by reversing route $R_k$, the route $R_k=(0,v_1^k,\ldots,v_{l_k-1}^k,0)$
becomes $(0,v_{l_k-1}^k,\ldots,v_1^k,0)$; the corresponding arc set
changes from $\{(v_i^k,v_{i+1}^k)\}$ to $\{(v_{i+1}^k,v_i^k)\}$, which
under the directed model is a \emph{different} arc set but encodes the
same undirected cycle, hence the same routing graph after symmetrization.
For the directed model with symmetric weights, both arc orientations are
identified with the same routing graph by the standard convention of
recording each cycle in its canonical (e.g., lex-smallest-first) orientation;
this is the symmetry quotient. Conversely, suppose
$\Phi(\mathcal{R})=\Phi(\mathcal{R}')=G_t$. The cycle decomposition of
$G_t$ (existence and uniqueness up to relabeling and reversal of
individual cycles is the standard Eulerian decomposition result for
graphs with prescribed in/out-degrees) shows that any two pre-images of
$G_t$ differ only by route-order and reversal.

The optimization equivalence of \cref{thm:main} respects this quotient:
since both the cost $C(\mathcal{R})$ and the routing graph $\Phi(\mathcal{R})$
are constant on each equivalence class, optimizing over $\F_{\mathrm{VRP}}$
and optimizing over $\F$ yield identical optimal values and identical
optimal classes.
\end{proof}

\begin{figure}[t]
\centering
\centering
\begin{tikzpicture}[
    node distance = 1.5cm,
    depot/.style = {circle, draw, fill=blue!20, minimum size=0.8cm},
    customer/.style = {circle, draw, fill=green!20, minimum size=0.8cm},
    edge/.style = {draw, -},
    selected/.style = {draw, ultra thick, red},
    deleted/.style = {draw, dashed, gray!50, thick},
    every node/.style = {font=\small}
]

\begin{scope}[xshift=0cm, yshift=0cm]
    \node[depot] (d1) at (0,2) {0};
    \node[customer] (c1) at (-1.5,0) {1};
    \node[customer] (c2) at (0,0.5) {2};
    \node[customer] (c3) at (1.5,-0.5) {3};
    
    \draw[gray!30, dashed] (d1) -- (c1);
    \draw[gray!30, dashed] (d1) -- (c2);
    \draw[gray!30, dashed] (d1) -- (c3);
    \draw[gray!30, dashed] (c1) -- (c2);
    \draw[gray!30, dashed] (c2) -- (c3);
    \draw[gray!30, dashed] (c1) -- (c3);
    
    \node[below] at (0,-0.8) {$G_s$: Complete Graph};
    \node[below] at (0,-1.2) {(Source Graph)};
\end{scope}

\draw[->, thick] (1.5,1) -- (3,1) node[midway, above] {Edit Operations};

\begin{scope}[xshift=5cm, yshift=0cm]
    \node[depot] (d2) at (0,2) {0};
    \node[customer] (c4) at (-1.5,0) {1};
    \node[customer] (c5) at (0,0) {2};
    \node[customer] (c6) at (1.5,0) {3};
    
    \draw[selected] (d2) -- (c4) node[midway, left] {$w_{01}$};
    \draw[selected] (c4) -- (c5) node[midway, above] {$w_{12}$};
    \draw[selected] (c5) -- (c6) node[midway, above] {$w_{23}$};
    \draw[selected] (c6) -- (d2) node[midway, right] {$w_{30}$};
    
    \node[below] at (0,-0.8) {$G_t$: Routing Graph};
    \node[below] at (0,-1.2) {(Target Graph)};
\end{scope}

\begin{scope}[xshift=2.5cm, yshift=-3cm]
    \node[depot] (d3) at (1,1) {0};
    \node[customer] (c7) at (-1,-1) {1};
    \node[customer] (c8) at (1,-0.75) {2};
    \node[customer] (c9) at (3,-2) {3};
    
    \draw[deleted] (d3) -- (c8) node[midway, right] {$c_{\text{del}}=w_{02}$};
    \draw[deleted] (d3) -- (c9) node[midway, right] {$c_{\text{del}}=w_{03}$};
    \draw[deleted] (c7) -- (c9) node[midway, below] {$c_{\text{del}}=w_{13}$};
    
    \draw[gray!20] (d3) -- (c7);
    \draw[gray!20] (c7) -- (c8);
    \draw[gray!20] (c8) -- (c9);
    \draw[gray!20] (c9) -- (d3);
    
    \node[below] at (0,-2.5) {Edit Costs: $c_{\text{del}}(e) = w(e)$, $c_{\text{ins}}(e) = 0$};
\end{scope}
\node at (4, -6.9) {$\GED(G_s, G_t) = \displaystyle\sum_{e \in E \setminus E_t} w(e) = \wt - \displaystyle\sum_{e \in E_t} w(e)$};

\end{tikzpicture}

\caption{VRP--GED mapping.  Kept edges (blue, bold) form the route;
deleted edges (grey, dashed) contribute to GED.  The sum of deleted
edge weights equals $\wt-C(\mathcal{R})$.
For visual clarity each pair of arcs $(i,j),(j,i)$ in $G_s$ is drawn
as a single undirected stroke; the formal directed model of
\cref{def:feasible-characterization} contains both arcs of every such
pair, and the routing graph uses each traversed pair in exactly one
direction.}
\label{fig:vpr-ged-mapping}
\end{figure}
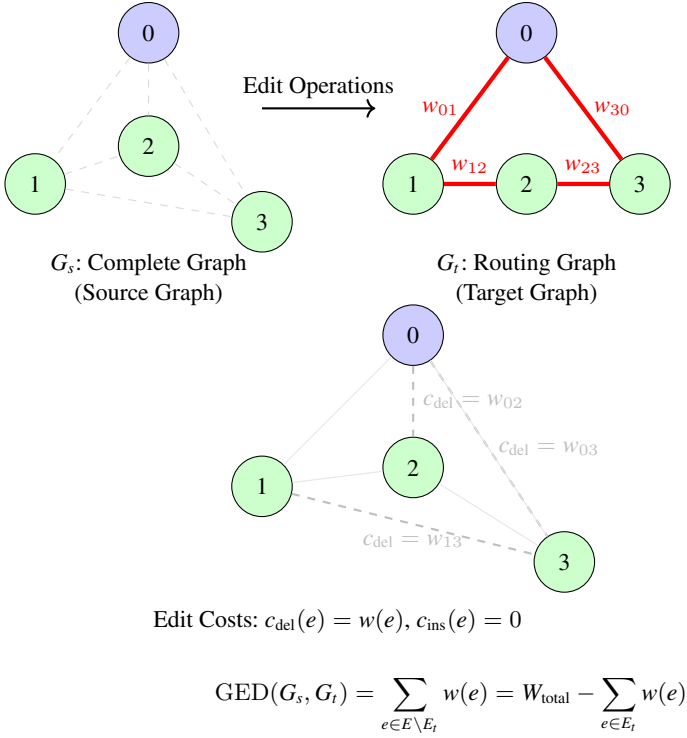

\subsection{Edit Cost Mapping and VRP-GED Objective Function Relation}\label{sec:cost-mapping}

The key element enabling the equivalence is the flexibility of assigning GED edit costs, which can be adapted to different VRP variants.  To transform the complete graph \(G_s\) into a feasible routing graph \(G_t\), we must delete every edge that is present in \(G_s\) but not in \(G_t\). Since \(E_t \subseteq E\), there are no insertions:

For edges \(e \in E\):
\begin{align}
c_{\text{del}}(e) &= w(e) \quad \text{for } e \in E \setminus E_t \label{eq:cost-del}\\
c_{\text{ins}}(e) &= 0 \quad \text{for } e \in E_t \setminus E \label{eq:cost-ins}
\end{align}

The deletion cost equals the original edge weight, penalizing the omission of an edge.

\begin{lemma}[Simplified GED for VRP]\label{lem:ged-simple}
With cost functions \eqref{eq:cost-del} and \eqref{eq:cost-ins}, the GED between \(G_s\) and \(G_t\) simplifies to:
\begin{equation}
\GED(G_s, G_t) = \sum_{e \in E \setminus E_t} w(e) \label{eq:ged-simple-vrp}
\end{equation}
\end{lemma}

\begin{proof}
Since \(E_t \subseteq E\) by definition, we have \(E_t \setminus E = \emptyset\), so the insertion term in \eqref{eq:ged-simple} vanishes. The deletion term becomes \(\sum_{e \in E \setminus E_t} w(e)\) by definition of \(c_{\text{del}}\) in \eqref{eq:cost-del}.
\end{proof}

As for objective relationship, the VRP objective \eqref{eq:vrp-objective} is the sum of the weights of edges actually traversed:
\begin{equation}
C(\mathcal{R}) = \sum_{e \in E_t} w(e) \label{eq:vrp-edge}
\end{equation}

Because \(G_s\) is complete, the total weight of all directed edges is constant:
\begin{equation}
\wt = \sum_{e \in E} w(e) = \sum_{e \in E_t} w(e) + \sum_{e \in E \setminus E_t} w(e) \label{eq:total-weight}
\end{equation}

From \eqref{eq:ged-simple-vrp} and \eqref{eq:total-weight}, we obtain the fundamental relation:
\begin{equation}
C(\mathcal{R}) = \wt - \GED(G_s, G_t) \label{eq:core-relation}
\end{equation}

Thus, minimizing VRP objective function (\(C(\mathcal{R})\)) is equivalent to maximizing \(\GED(G_s, G_t)\). Since \(\wt\) is constant, the VRP reduces to finding a feasible routing graph that maximizes the sum of weights of the edges that are \emph{not} used, as depicted in Figure~\ref{fig:vpr-ged-mapping}.

\section{Main Theorem and Proof}\label{sec:proof}
\begin{theorem}[VRP-GED Equivalence]\label{thm:main}
Let \(\I = (V,E,w,m,Q,D)\) be a VRP instance with complete graph \(G_s = (V,E,w)\). Let \(\F\) be the set of feasible routing graphs corresponding to feasible VRP solutions. Define GED with costs as in \eqref{eq:cost-del} and \eqref{eq:cost-ins}. Then:
\begin{equation}
\min_{\mathcal{R} \in \F} C(\mathcal{R}) = \wt - \max_{G_t \in \F} \GED(G_s, G_t) \label{eq:main-theorem}
\end{equation}
Equivalently, the optimization problems are identical up to an additive constant and sign change.
\end{theorem}

\begin{proof}
The proof proceeds in three parts:

\textbf{Part 1: Solution Correspondence.} By Lemma~\ref{lem:correspondence}, there exists a cost-preserving correspondence between the set of feasible VRP solutions (modulo route order and reversal) and the set of feasible routing graphs. This establishes that optimizing over one set is equivalent to optimizing over the other.

\textbf{Part 2: Objective Transformation.} For any \(G_t \in \F\) with corresponding solution \(\mathcal{R} = \Phi^{-1}(G_t)\), Lemma~\ref{lem:ged-simple} gives:
\begin{align}
\GED(G_s, G_t) &= \sum_{e \in E \setminus E_t} w(e) \label{eq:proof-1}
\end{align}
From \eqref{eq:total-weight}, we have:
\begin{align}
\sum_{e \in E_t} w(e) &= \wt - \sum_{e \in E \setminus E_t} w(e) = \wt - \GED(G_s, G_t) \label{eq:proof-2}
\end{align}
But \(\sum_{e \in E_t} w(e) = C(\mathcal{R})\) by \eqref{eq:vrp-edge}. Therefore:
\begin{align}
C(\mathcal{R}) = \wt - \GED(G_s, G_t) \label{eq:proof-3}
\end{align}

\textbf{Part 3: Optimization Equivalence.} Since \(\wt\) is constant with respect to the optimization:
\begin{align}
\min_{\mathcal{R} \in \F} C(\mathcal{R}) &= \min_{\mathcal{R} \in \F} (\wt - \GED(G_s, \Phi(\mathcal{R}))) \nonumber \\
&= \wt - \max_{\mathcal{R} \in \F} \GED(G_s, \Phi(\mathcal{R})) \nonumber \\
&= \wt - \max_{G_t \in \F} \GED(G_s, G_t) \label{eq:proof-4}
\end{align}
where the last equality uses the correspondence from Lemma~\ref{lem:correspondence} to replace optimization over \(\mathcal{R}\) with optimization over \(G_t\).
\end{proof}

\noindent
In summary, the theorem is established through a correspondence between VRP solutions and feasible routing graphs (Part 1), a transformation linking the VRP cost to the GED via the constant total edge weight (Part 2), and the affine relationship that turns minimization of cost into maximization of GED (Part 3).

The value of \cref{thm:main} lies not in the algebraic identity, but in the structural change it induces.  The classical VRP objective $C(\mathcal{R}) = \sum_k \sum_i w(v_i^k, v_{i+1}^k)$ is \emph{path-aggregated}: each edge's contribution to cost depends on its position within a route sequence, and the objective cannot be decomposed over individual edges without reference to the route structure.  In contrast, the GED objective $\sum_{e \in E \setminus E_t} w(e)$ is \emph{edge-additive}: each edge contributes independently, regardless of route context.  This enables three capabilities that require non-trivial post-processing outside the standard VRP framework:
\begin{enumerate}[nosep]
\item \emph{Per-edge quality attribution}: determining each edge's individual contribution to the optimality gap (\cref{sec:exp-decomp}).
\item \emph{Per-edge supervision for learning}: defining a well-posed loss function over individual edges for GNN training (\cref{sec:gnn}).
\item \emph{Approximation transfer}: translating edge-level approximation guarantees into cost-level guarantees via \cref{thm:approximation}.
\end{enumerate}

\begin{corollary}[Solution Set Equivalence]\label{cor:solution-set}
The optimal solution sets are identical under the correspondence \(\Phi\):
\begin{equation}
\Phi\left(\argmin_{\mathcal{R} \in \F} C(\mathcal{R})\right) = \argmax_{G_t \in \F} \GED(G_s, G_t)
\end{equation}
\end{corollary}

Since the transformation between VRP and GED formulations is polynomial-time (requiring only computation of \(\wt\) in \(O(n^2)\) time), and standard VRP is NP-hard, the GED formulation with costs \eqref{eq:cost-del} and \eqref{eq:cost-ins} is also NP-hard.

\section{Extended Theoretical Results}\label{sec:extended}

\begin{theorem}[Merge $\Delta$GED]\label{thm:merge}
Merging routes \(R_a=[\ldots,i,0]\) and \(R_b=[0,j,\ldots]\) into
\([\ldots,i,j,\ldots]\) increases GED by the Clarke--Wright saving:
\begin{equation}\label{eq:merge}
\Delta\GED(i,j)=w(i,0)+w(0,j)-w(i,j)=s(i,j).
\end{equation}
(Note: The additive decomposition claim \(\GED(\mathcal{R})=\sum_k\GED(R_k)\) requires careful definition of route-level GED and is omitted; only the local merge identity is kept.)
\end{theorem}

\begin{proof}
After the merge, arcs \((i,0)\) and \((0,j)\) are deleted and \((i,j)\)
is kept; the net change in deleted weight is \(s(i,j)\).
\end{proof}

\cref{thm:merge} reveals that CW savings are exactly per-merge GED
increments: processing merges in descending savings order is
\emph{greedy GED maximisation}.

\begin{heuristic}[Route Elimination Priority]\label{cor:elim}
For each route $R_k$, define the net-GED score
$\mathrm{score}(R_k)=C(R_k)-\sum_{c\in R_k}\delta^*(c)$,
where $\delta^*(c)=\min_{R_j\neq R_k,\,\mathrm{feasible}}\Delta C(c\to R_j)$.
When excess routes must be eliminated, prioritize the route with the
highest score. This is a \emph{heuristic rule motivated by GED
semantics}, not a derived mathematical consequence; it does not
guarantee optimality but is effective in practice.
\end{heuristic}

\begin{corollary}[TSP Special Case]\label{cor:tsp}
When \(m=1\), \(\min_{\text{Hamiltonian cycle}}C(R)=\wt-\max_R\GED(G_s,G_R)\),
establishing GED as a dual TSP formulation.
\end{corollary}

\begin{proposition}[Uncapacitated VRP]\label{prop:uncapacitated}
When capacity constraints are relaxed (\(Q \geq \sum_{i=1}^n d_i\)), the set \(\F\) of feasible routing graphs consists of all subgraphs satisfying degree constraints \(\degr^+(0)=\degr^-(0)=m\) and \(\degr^+(i)=\degr^-(i)=1\) for all customers \(i\), with exactly \(m\) edge-disjoint directed cycles each containing the depot.
\end{proposition}

\subsection{Structural Properties of Optimal Solutions}

We define the complement notation $\overline{G_t}=(V,E\setminus E_t)$ for any feasible routing graph $G_t$; this notation is used throughout the remainder of the paper.

\label{rem:complement}
Let $G_t^*$ be an optimal routing graph maximizing $\GED(G_s,G_t)$.
Then its complement $\overline{G_t^*}=(V,E\setminus E_t^*)$ has
\emph{maximum total edge weight} among complements of feasible routing
graphs. This is an immediate algebraic consequence of \cref{thm:main}:
since $\sum_{e\in E_t}w(e)=\wt-\sum_{e\in E\setminus E_t}w(e)$ and
$\wt$ is constant, minimizing the kept weight is equivalent to
maximizing the deleted weight. We record it separately because the
complement view --- optimizing over \emph{unused} rather than used
edges --- motivates the analyses of \cref{sec:exp-decomp,sec:freq}.

\section{Connections to Classical Graph Problems}\label{sec:connections}

\subsection{GED Approximation Transfer}

\begin{theorem}[Approximation Transfer]\label{thm:approximation}
If there exists a polynomial-time algorithm that computes a solution \(\hat{G}_t\) such that:
\begin{equation}
\GED(G_s, \hat{G}_t) \geq \alpha \cdot \max_{G_t \in \F} \GED(G_s, G_t)
\end{equation}
for some \(\alpha \leq 1\), then the corresponding VRP solution \(\hat{\mathcal{R}} = \Phi^{-1}(\hat{G}_t)\) satisfies:
\begin{equation}
C(\hat{\mathcal{R}}) \leq \wt - \alpha \cdot (\wt - C(\mathcal{R}^*)) = (1-\alpha)\wt + \alpha C(\mathcal{R}^*)
\end{equation}
where \(\mathcal{R}^*\) is the optimal VRP solution.
\end{theorem}
\noindent

Because \(\wt\) grows quadratically with \(n\), a GED ratio \(\alpha\) close to 1 does not guarantee a small absolute cost gap; the bound should be interpreted as a theoretical transfer rather than a practical guarantee.

\begin{proof}
From \eqref{eq:core-relation}, \(C(\hat{\mathcal{R}}) = \wt - \GED(G_s, \hat{G}_t)\). Using the approximation guarantee:
\begin{align}
C(\hat{\mathcal{R}}) &\leq \wt - \alpha \cdot \max_{G_t \in \F} \GED(G_s, G_t) \nonumber \\
&= \wt - \alpha \cdot (\wt - C(\mathcal{R}^*)) \nonumber \\
&= (1-\alpha)\wt + \alpha C(\mathcal{R}^*)
\end{align}
\end{proof}

\begin{corollary}[Constant Factor Approximation]\label{cor:approx}
If \(\alpha > 1 - \frac{1}{\rho}\) for some \(\rho > 1\), then \(C(\hat{\mathcal{R}}) \leq \rho C(\mathcal{R}^*)\) when \(\wt\) is bounded relative to \(C(\mathcal{R}^*)\).
\end{corollary}

Note that the practical value of \cref{thm:approximation} lies not in the absolute tightness of $\alpha$ but in establishing a formal mechanism for transferring GED approximation guarantees to VRP cost bounds.

\subsection{GED-Based Multi-Start Heuristic}\label{sec:solver}

To generate heuristic solutions for the structural analysis in
\cref{sec:experiments}, we use a standard Clarke--Wright multi-start
heuristic with feasibility repair.  The heuristic is not a
contribution of this paper; we describe it briefly to ensure
reproducibility.
\cref{tab:ged-map} maps each component to its GED interpretation.

\begin{table}[t]
\centering
\renewcommand{\arraystretch}{1.2}
\caption{GED Interpretation of Each heuristic Component.}
\label{tab:ged-map}
\begin{tabularx}{\columnwidth}{lX}
\toprule
\textbf{Component} & \textbf{GED interpretation} \\
\midrule
CW merge           & $\Delta\GED=s(i,j)$: greedy maximisation (\cref{thm:merge}) \\
2-opt reversal     & Cost $\downarrow$ $\Rightarrow$ GED $\uparrow$ (\cref{thm:main}) \\
Route elimination  & Ordered by net-GED score (heuristic rule, \cref{cor:elim}) \\
T1 direct insert   & Cheapest insertion $=$ max-$\Delta$GED placement (\cref{thm:merge}) \\
T2/T3 ejection     & Capacity-feasible chain preserving GED gains \\
Noisy restarts     & Independent exploration of the GED landscape \\
\bottomrule
\end{tabularx}
\end{table}

\subsection{GED-Greedy Construction}\label{sec:construction}

\subsubsection{Clarke--Wright as Greedy GED Maximization}
The heuristic initializes $n$ singleton routes $\{[0,v_i,0]\}$.
By \cref{thm:merge}, the Clarke--Wright savings value
$s(i,j) = w(i,0) + w(0,j) - w(i,j)$
equals the GED increment from merging routes through edge $(i,j)$.
Processing merges in descending savings order is therefore
greedy GED maximisation: each step selects the feasible merge
that increases $\GED(G_s, G_t)$ the most.

After the merge loop, intra-route 2-opt improves each route.
Every accepted reversal strictly decreases $C(R)$ and therefore
strictly increases GED (\cref{thm:main}).

\subsubsection{Noisy Restarts}
A single greedy merge sequence may commit to suboptimal early
merges that cannot be undone.  To explore diverse regions of the
GED landscape, we generate multiple restarts by perturbing the
savings values:
\begin{equation}\label{eq:noisy}
\tilde{s}(i,j)=s(i,j)+\varepsilon_{ij},\quad
\varepsilon_{ij}\sim\mathcal{N}\!\left(0,\,\sigma^2 s(i,j)^2\right),
\end{equation}
with $\sigma\in[0,\,0.5]$ growing across restarts.  Different noise
realisations produce different merge orderings and therefore
structurally distinct solutions: on an $n{=}199$ instance, two
noisy restarts share only $28\%$ of their edges on average.

Each feasible result is 2-opt polished and compared against the
incumbent.  The solver returns the best solution found across all
restarts.

\subsection{Feasibility Repair: T1/T2/T3 Cascade}\label{sec:repair}

When CW produces $m'>m$ routes, surplus routes must be eliminated
while preserving feasibility.  The cascade eliminates one route at a
time, always selecting the route with the highest \emph{net-GED score}
(\cref{cor:elim}):
\begin{equation}\label{eq:netscore}
\mathrm{score}(R_k) = C(R_k) - \sum_{c\in R_k} \delta^*(c),
\end{equation}
where $\delta^*(c) = \min_{R_j \neq R_k,\,\text{feasible}} \Delta C(c \to R_j)$
is the cheapest feasible reinsertion cost for customer~$c$.
The route with the highest score has the most ``wasted'' GED as its
route cost far exceeds the cost of redistributing its customers, and
is therefore eliminated first.

Once a route is selected for elimination, its customers are
reinserted in decreasing order of demand (heaviest first, to avoid
capacity deadlocks) using a three-tier strategy:

\begin{description}
\item[T1 — Direct insertion.]
Insert customer~$c$ at the cheapest feasible position in any
remaining route.  The insertion cost equals $-\Delta\GED$ for that
position (\cref{thm:merge}), so cheapest insertion is equivalent to
maximum-$\Delta$GED placement.

\item[T2 — One-level ejection.]
When all routes lack sufficient residual capacity for~$c$:
identify a route~$R_j$ where ejecting one customer~$d$ to a third
route~$R_k$ creates room for~$c$.  The chain $d \to R_k$,
$c \to R_j$ is executed if the combined reinsertion is
capacity-feasible.

\item[T3 — Two-level ejection.]
For instances with $>97\%$ vehicle utilisation, even T2 may fail
because every route is near capacity.  T3 extends the chain by
first ejecting a customer~$e$ from~$R_k$ to a fourth route~$R_l$,
creating room for~$d$, which in turn creates room for~$c$:
$e \to R_l$, $d \to R_k$, $c \to R_j$.
\end{description}

If the selected route cannot be eliminated by any tier, it is
skipped and the next-highest-scoring route is attempted.
When $m' < m$ (too few routes), the solver splits the route with the
best-cost feasible midpoint partition.

\Cref{alg:solver} summarizes the complete solver.

\begin{algorithm}[t]
\caption{Multi-Start GED-Greedy Construction}
\label{alg:solver}
\begin{algorithmic}[1]
\Require VRP instance, time budget $T$
\Ensure Best feasible solution $R^*$
\State $R^* \leftarrow \textsc{nil}$;\; $C^* \leftarrow \infty$
\For{$\mathrm{attempt}=0,1,2,\ldots$}
  \If{elapsed $> T$} \textbf{break} \EndIf
  \State $\sigma\leftarrow0$ if attempt$\,{=}\,0$, else
         $\min(0.02+0.003\cdot\mathrm{attempt},\;0.5)$
  \State $R\leftarrow\textsc{CW}(\sigma)$
         \Comment{greedy GED max., \cref{thm:merge}}
  \State $R\leftarrow\textsc{2-Opt}(R)$
         \Comment{$\Delta C<0 \Rightarrow \Delta\GED>0$}
  \If{$|R| > m$}
    \State $R\leftarrow\textsc{T1/T2/T3}(R)$
           \Comment{eliminate by net-GED score, \cref{cor:elim}}
  \ElsIf{$|R| < m$}
    \State $R\leftarrow\textsc{Split}(R)$
  \EndIf
  \If{$R$ is feasible \textbf{and} $C(R) < C^*$}
    \State $R^* \leftarrow R$;\; $C^* \leftarrow C(R)$
  \EndIf
\EndFor
\State \Return $R^*$
\end{algorithmic}
\end{algorithm}

\subsection{Complexity}

Each restart requires $O(n^2\log n)$ for the savings sort and
$O(n^2)$ for the T1/T2/T3 cascade.  Intra-route 2-opt is $O(n^2)$
per route.  The number of restarts is bounded by the time budget;
for $n{\leq}200$ at a 1-second budget, the solver completes
20--80 restarts depending on utilisation.

\section{Experimental Validation}\label{sec:experiments}

We validate the theoretical results and conduct structural analysis on
90 CVRP instances from eight standard benchmark families
(Augerat~A/B/P, Christofides~E, Fisher~F, Meisel~M, TSPLIB) with
proven optimal solutions.  For each instance, we extract the optimal
routing graph~$G_t^*$ from published solution files (CVRPLIB~\cite{cvrplib})
and compute all GED quantities independently.  Heuristic solutions are
generated by 50 noisy Clarke--Wright restarts per instance with 2-opt
polish.

\Cref{tab:instances} summarizes the experimental setup for reproducibility.
Distances are rounded to the nearest integer as specified in each benchmark
family.  The same 50 fixed random seeds (integers $0, 1, \ldots, 49$) are
used across all instances.

\begin{table}[t]
\centering
\caption{Experimental setup: benchmark families, instance counts, and sizes.
Implementation: Python 3.10; Intel i7-12700H (2.30\,GHz), 16\,GB RAM;
random seeds $0$--$49$; 1-second time budget per restart; 50 restarts per instance.
Optimal solutions loaded from CVRPLIB~\cite{cvrplib}.
Code will be released upon acceptance.}
\label{tab:instances}
\renewcommand{\arraystretch}{1.15}
\begin{tabular}{@{}llrr@{}}
\toprule
Family & Source & \#Inst.\ & $n$ range \\
\midrule
Augerat A      & \cite{augerat1995}      & 27 & 16--80  \\
Augerat B      & \cite{augerat1995}      & 23 & 31--78  \\
Augerat P      & \cite{augerat1995}      & 24 & 15--101 \\
Christofides E & \cite{christofides1969} & 10 & 22--100 \\
Fisher F       & \cite{fisher1994}       &  2 & 70--134 \\
Meisel M       & \cite{meisel2011}       &  2 & 29--31  \\
TSPLIB         & \cite{reinelt1991}      &  2 & 134--199\\
\midrule
\textbf{Total} &                         & \textbf{90} & \textbf{15--199} \\
\bottomrule
\end{tabular}
\end{table}

\subsection{Theorem Verification}\label{sec:exp-verify}

The identity $C(R) = W_{\mathrm{total}} - \GED(G_s, G_t)$
(\cref{thm:main}) holds to zero error on all 90 instances.
The per-merge identity $\Delta\GED = s(i,j)$ (\cref{thm:merge}),
verified by instrumenting the CW merge loop, also holds to zero
error across all merges on all instances.  These verifications
confirm computational correctness; the remainder of this section
presents the structural analyses that the framework enables.

\subsection{Optimal GED Structure}\label{sec:exp-structure}

\cref{tab:ged-structure} characterises the structure of optimal
routing graphs.  The ratio $\GED^*/W_{\mathrm{total}}$ measures
what fraction of the complete graph's total edge weight is
``edited away'' (deleted) in the optimal solution.

\begin{table}[h!]
\centering
\caption{Optimal GED structure across 90 benchmark instances.}
\label{tab:ged-structure}
\begin{tabular}{lcccc}
\toprule
\textbf{Metric} & \textbf{Min} & \textbf{Median} & \textbf{Mean} & \textbf{Max} \\
\midrule
$\GED^*/W_{\mathrm{total}}$     & 83.0\% & 98.3\% & 97.7\% & 99.7\% \\
$|E_t^*|/|E|$                    & 1.5\%  &  4.8\% &  5.5\% & 17.1\% \\
\bottomrule
\end{tabular}
\end{table}

Under the directed model of \cref{def:feasible-characterization},
every feasible routing graph uses exactly $|E_t|=n+m$ arcs:
each customer has out-degree $1$, the depot has out-degree $m$, and
$|E_t|$ equals the sum of out-degrees. This is a counting consequence
of feasibility, not a property of optimality. Because edge weights are
symmetric ($w(i,j)=w(j,i)$), throughout the experimental section we
report $|E_t|$ and $|E|$ as counts of \emph{unordered} edge pairs --- the
standard CVRP convention --- so $|E|=n(n{+}1)/2$ and the reported
ratio in \cref{tab:ged-structure} equals $2(n+m)/n(n{+}1)$. Directed
and undirected counts differ by a factor of two; the ratio is identical
up to that factor.

The 5.5\% mean ratio is therefore a counting fact about all feasible
CVRP solutions on the tested instance sizes ($n=15$ to $n=199$), not a
distinguishing property of optimal solutions. The substantive
structural question --- which the GED reformulation enables --- is not
\emph{how many} edges are used but \emph{which}: how predictable the
identity of optimal edges is across heuristic restarts, and how the
optimality gap distributes across edges that are kept versus omitted.
\cref{tab:overlap} below addresses the first question; \cref{sec:exp-decomp}
addresses the second.

\subsection{Edge Overlap: Heuristic vs.\ Optimal}\label{sec:exp-overlap}

For each optimal edge $e \in E_t^*$, we measure how often it
appears across 50 noisy CW restarts.  \cref{tab:overlap}
categorises optimal edges by their heuristic frequency.

\begin{table}[h!]
\centering
\caption{CW frequency of optimal edges across 90 instances
(4{,}986 optimal edges total).}
\label{tab:overlap}
\begin{tabular}{lrr}
\toprule
\textbf{CW frequency} & \textbf{Edges} & \textbf{Fraction} \\
\midrule
$\geq 80\%$ (always found)      & 436   &  8.7\% \\
$20$--$80\%$ (sometimes found)  & 3{,}799 & 76.2\% \\
$5$--$20\%$ (rarely found)      & 602   & 12.1\% \\
$< 5\%$ (never found in tested restarts)     & 149   &  3.0\% \\
\bottomrule
\end{tabular}
\end{table}

Only 8.7\% of optimal edges are consistently found by CW
($f \geq 0.8$) the stable GED backbone of
\cref{prop:freq}.  The vast majority (76.2\%) are found
intermittently, confirming that CW explores different merge
orderings but without systematic preference for optimal edges.
Most critically, 3.0\% of optimal edges were \textit{not found in any of the tested restarts} ($f < 0.05$).  We report this as an empirical observation, not a formal proof of unreachability.
The figure is specific to the Gaussian noise model of \cref{eq:noisy}
with $\sigma\in[0,0.5]$ and 50 restarts per instance; alternative
perturbation strategies (ruin-and-recreate, large-neighborhood search,
uniform-noise restarts) may reach a different subset and we make no
claim about reachability outside the tested protocol.

\subsection{Gap Decomposition}\label{sec:exp-decomp}

For every heuristic solution (not only the best), we decompose
the cost gap $C(R) - C(R^*)$ into two graph-structural components
enabled by the GED formulation:

\begin{itemize}
\item \textbf{Missed weight:}
$\displaystyle\sum_{e \in E_t^* \setminus E_t} w(e)$
--- the total weight of optimal edges the heuristic failed to use.

\item \textbf{Wasted weight:}
$\displaystyle\sum_{e \in E_t \setminus E_t^*} w(e)$
--- the total weight of non-optimal edges the heuristic used instead.
\end{itemize}

The gap satisfies
$C(R) - C(R^*) = \text{wasted} - \text{missed}$
exactly for every feasible solution, regardless of seed.
\cref{tab:decomp} reports both the best solution per instance
and the average across all 50 seeds.
%

\begin{figure*}[t]
\centering
\begin{tikzpicture}[
    >=Stealth,
    font=\small,
    depot/.style={circle, draw, fill=black, minimum size=5pt, inner sep=0pt},
    cust/.style={circle, draw, fill=blue!20, minimum size=3.5pt, inner sep=0pt},
    custmiss/.style={circle, draw, fill=red!40, minimum size=5pt, inner sep=0pt},
    gedge/.style={draw, gray!25, line width=0.15pt},        
    route1/.style={draw, blue!70, line width=1.2pt},
    route2/.style={draw, red!60!black, line width=1.2pt},
    route3/.style={draw, green!50!black, line width=1.2pt},
    route4/.style={draw, orange!80!black, line width=1.2pt},
    route5/.style={draw, violet!70, line width=1.2pt},
    shared/.style={draw, blue!40, line width=0.9pt},
    missed/.style={draw, red!70, line width=1.4pt, dashed},
    wasted/.style={draw, orange!80, line width=1.4pt, densely dotted},
]

\def\sc{0.72}

\begin{scope}[shift={(0,0)}, scale=\sc]
  \node[font=\small\bfseries, anchor=south] at (3.5, 7.8) {(a) Complete graph $G_s$};

  \coordinate (n0) at (1.04,1.34);
  \coordinate (n1) at (0.00,3.57);
  \coordinate (n2) at (6.40,1.79);
  \coordinate (n3) at (5.06,4.77);
  \coordinate (n4) at (6.85,6.70);
  \coordinate (n5) at (2.38,2.23);
  \coordinate (n6) at (5.21,4.47);
  \coordinate (n7) at (2.09,0.60);
  \coordinate (n8) at (6.85,0.45);
  \coordinate (n9) at (4.02,3.43);
  \coordinate (n10) at (1.64,0.89);
  \coordinate (n11) at (1.34,3.43);
  \coordinate (n12) at (4.17,4.62);
  \coordinate (n13) at (0.30,7.00);
  \coordinate (n14) at (4.77,3.13);
  \coordinate (n15) at (5.06,0.00);
  \coordinate (n16) at (0.15,1.79);
  \coordinate (n17) at (1.34,6.70);
  \coordinate (n18) at (1.49,5.96);
  \coordinate (n19) at (4.91,6.70);
  \coordinate (n20) at (2.98,1.64);
  \coordinate (n21) at (1.34,5.51);
  \coordinate (n22) at (1.04,5.81);
  \coordinate (n23) at (5.81,3.43);
  \coordinate (n24) at (1.34,4.77);
  \coordinate (n25) at (1.94,3.57);
  \coordinate (n26) at (2.09,1.19);
  \coordinate (n27) at (1.79,4.77);
  \coordinate (n28) at (4.32,3.72);
  \coordinate (n29) at (1.94,7.00);
  \coordinate (n30) at (1.49,6.70);
  \coordinate (n31) at (4.47,6.11);
  \coordinate (n32) at (1.04,6.11);
  \coordinate (n33) at (2.23,6.40);
  \coordinate (n34) at (5.21,2.98);
  \coordinate (n35) at (6.70,1.49);

  \foreach \i in {0,1,...,35} {
    \foreach \j in {\i,...,35} {
      \pgfmathparse{int(mod(\i*36+\j, 4))}
      \ifnum\pgfmathresult=0
        \draw[gedge] (n\i) -- (n\j);
      \fi
    }
  }

  \foreach \i in {1,...,35} { \node[cust] at (n\i) {}; }
  \node[depot] at (n0) {};

  \node[font=\scriptsize, anchor=north] at (3.5, -0.5)
    {$|E| = 630$ edges,\; $W_{\text{total}} = 32{,}608$};
\end{scope}

\begin{scope}[shift={(6.0,0)}, scale=\sc]
  \node[font=\small\bfseries, anchor=south] at (3.5, 7.8) {(b) Optimal $G_t^*$: 5 routes};

  \coordinate (n0) at (1.04,1.34);
  \coordinate (n1) at (0.00,3.57);
  \coordinate (n2) at (6.40,1.79);
  \coordinate (n3) at (5.06,4.77);
  \coordinate (n4) at (6.85,6.70);
  \coordinate (n5) at (2.38,2.23);
  \coordinate (n6) at (5.21,4.47);
  \coordinate (n7) at (2.09,0.60);
  \coordinate (n8) at (6.85,0.45);
  \coordinate (n9) at (4.02,3.43);
  \coordinate (n10) at (1.64,0.89);
  \coordinate (n11) at (1.34,3.43);
  \coordinate (n12) at (4.17,4.62);
  \coordinate (n13) at (0.30,7.00);
  \coordinate (n14) at (4.77,3.13);
  \coordinate (n15) at (5.06,0.00);
  \coordinate (n16) at (0.15,1.79);
  \coordinate (n17) at (1.34,6.70);
  \coordinate (n18) at (1.49,5.96);
  \coordinate (n19) at (4.91,6.70);
  \coordinate (n20) at (2.98,1.64);
  \coordinate (n21) at (1.34,5.51);
  \coordinate (n22) at (1.04,5.81);
  \coordinate (n23) at (5.81,3.43);
  \coordinate (n24) at (1.34,4.77);
  \coordinate (n25) at (1.94,3.57);
  \coordinate (n26) at (2.09,1.19);
  \coordinate (n27) at (1.79,4.77);
  \coordinate (n28) at (4.32,3.72);
  \coordinate (n29) at (1.94,7.00);
  \coordinate (n30) at (1.49,6.70);
  \coordinate (n31) at (4.47,6.11);
  \coordinate (n32) at (1.04,6.11);
  \coordinate (n33) at (2.23,6.40);
  \coordinate (n34) at (5.21,2.98);
  \coordinate (n35) at (6.70,1.49);

  \draw[route1] (n0)--(n9)--(n6)--(n3)--(n4)--(n19)--(n31)--(n12)--(n0);
  \draw[route2] (n0)--(n28)--(n14)--(n34)--(n23)--(n2)--(n35)--(n8)--(n15)--(n0);
  \draw[route3] (n0)--(n16)--(n11)--(n24)--(n27)--(n25)--(n5)--(n20)--(n0);
  \draw[route4] (n0)--(n10)--(n7)--(n26)--(n0);
  \draw[route5] (n0)--(n1)--(n22)--(n32)--(n13)--(n17)--(n30)--(n29)--(n33)--(n18)--(n21)--(n0);

  \foreach \i in {1,...,35} { \node[cust] at (n\i) {}; }
  \node[depot] at (n0) {};

  \node[font=\scriptsize, anchor=north] at (3.5, -0.5)
    {$|E_t^*| = 40$ edges,\; $C^* = 799$};
\end{scope}

\begin{scope}[shift={(12.0,0)}, scale=\sc]
  \node[font=\small\bfseries, anchor=south] at (3.5, 7.8) {(c) Heuristic: missed \& wasted};

  \coordinate (n0) at (1.04,1.34);
  \coordinate (n1) at (0.00,3.57);
  \coordinate (n2) at (6.40,1.79);
  \coordinate (n3) at (5.06,4.77);
  \coordinate (n4) at (6.85,6.70);
  \coordinate (n5) at (2.38,2.23);
  \coordinate (n6) at (5.21,4.47);
  \coordinate (n7) at (2.09,0.60);
  \coordinate (n8) at (6.85,0.45);
  \coordinate (n9) at (4.02,3.43);
  \coordinate (n10) at (1.64,0.89);
  \coordinate (n11) at (1.34,3.43);
  \coordinate (n12) at (4.17,4.62);
  \coordinate (n13) at (0.30,7.00);
  \coordinate (n14) at (4.77,3.13);
  \coordinate (n15) at (5.06,0.00);
  \coordinate (n16) at (0.15,1.79);
  \coordinate (n17) at (1.34,6.70);
  \coordinate (n18) at (1.49,5.96);
  \coordinate (n19) at (4.91,6.70);
  \coordinate (n20) at (2.98,1.64);
  \coordinate (n21) at (1.34,5.51);
  \coordinate (n22) at (1.04,5.81);
  \coordinate (n23) at (5.81,3.43);
  \coordinate (n24) at (1.34,4.77);
  \coordinate (n25) at (1.94,3.57);
  \coordinate (n26) at (2.09,1.19);
  \coordinate (n27) at (1.79,4.77);
  \coordinate (n28) at (4.32,3.72);
  \coordinate (n29) at (1.94,7.00);
  \coordinate (n30) at (1.49,6.70);
  \coordinate (n31) at (4.47,6.11);
  \coordinate (n32) at (1.04,6.11);
  \coordinate (n33) at (2.23,6.40);
  \coordinate (n34) at (5.21,2.98);
  \coordinate (n35) at (6.70,1.49);

  \draw[shared] (n0)--(n1);  \draw[shared] (n0)--(n9);
  \draw[shared] (n0)--(n10); \draw[shared] (n0)--(n15);
  \draw[shared] (n0)--(n16); \draw[shared] (n0)--(n20);
  \draw[shared] (n0)--(n21); \draw[shared] (n0)--(n26);
  \draw[shared] (n1)--(n22); \draw[shared] (n2)--(n35);
  \draw[shared] (n3)--(n4);  \draw[shared] (n3)--(n6);
  \draw[shared] (n4)--(n19); \draw[shared] (n5)--(n20);
  \draw[shared] (n5)--(n25); \draw[shared] (n7)--(n10);
  \draw[shared] (n7)--(n26); \draw[shared] (n8)--(n15);
  \draw[shared] (n8)--(n35); \draw[shared] (n11)--(n16);
  \draw[shared] (n11)--(n24);\draw[shared] (n13)--(n17);
  \draw[shared] (n13)--(n32);\draw[shared] (n14)--(n34);
  \draw[shared] (n17)--(n30);\draw[shared] (n18)--(n21);
  \draw[shared] (n18)--(n33);\draw[shared] (n19)--(n31);
  \draw[shared] (n22)--(n32);\draw[shared] (n23)--(n34);
  \draw[shared] (n24)--(n27);\draw[shared] (n25)--(n27);
  \draw[shared] (n29)--(n30);\draw[shared] (n29)--(n33);

  \draw[missed] (n0)--(n12)  node[pos=0.5, font=\tiny, fill=white, inner sep=0.5pt] {61};
  \draw[missed] (n0)--(n28)  node[pos=0.4, font=\tiny, fill=white, inner sep=0.5pt] {54};
  \draw[missed] (n2)--(n23)  node[pos=0.5, font=\tiny, fill=white, inner sep=0.5pt] {23};
  \draw[missed] (n6)--(n9)   node[pos=0.5, font=\tiny, fill=white, inner sep=0.5pt] {21};
  \draw[missed] (n12)--(n31) node[pos=0.5, font=\tiny, fill=white, inner sep=0.5pt] {20};
  \draw[missed] (n14)--(n28) node[pos=0.5, font=\tiny, fill=white, inner sep=0.5pt] {10};

  \draw[wasted] (n0)--(n14)  node[pos=0.6, font=\tiny, fill=white, inner sep=0.5pt] {55};
  \draw[wasted] (n0)--(n31)  node[pos=0.3, font=\tiny, fill=white, inner sep=0.5pt] {79};
  \draw[wasted] (n2)--(n12)  node[pos=0.5, font=\tiny, fill=white, inner sep=0.5pt] {48};
  \draw[wasted] (n6)--(n23)  node[pos=0.5, font=\tiny, fill=white, inner sep=0.5pt] {16};
  \draw[wasted] (n9)--(n28)  node[pos=0.5, font=\tiny, fill=white, inner sep=0.5pt] {6};
  \draw[wasted] (n12)--(n28) node[pos=0.5, font=\tiny, fill=white, inner sep=0.5pt] {12};

  \foreach \i in {1,...,35} { \node[cust] at (n\i) {}; }
  \node[depot] at (n0) {};

  \node[font=\scriptsize, anchor=north, align=center] at (3.5, -0.5)
    {34 shared, 6 missed (wt\,189), 6 wasted (wt\,216)};
\end{scope}

\node[draw, thick, rounded corners=4pt, fill=yellow!8,
      font=\small, align=center, anchor=north]
  at (9.6, -1.3)
  {Theorem~1:\;
   $\underbrace{C(\mathcal{R}) = 826}_{\text{heuristic cost}}
    \;=\;
    \underbrace{W_{\text{total}} = 32{,}608}_{\text{constant}}
    \;-\;
    \underbrace{\GED(G_s, G_t) = 31{,}782}_{\sum \text{deleted edges}}$
   \qquad
   Gap $= \underbrace{216}_{\text{wasted}} - \underbrace{189}_{\text{missed}} = 27$};

\node[anchor=north west, font=\scriptsize] at (0.0, -1.0) {%
  \begin{tikzpicture}[baseline=-0.5ex]
    \draw[shared, line width=1.5pt] (0,0) -- (0.5,0);
    \node[anchor=west] at (0.6,0) {shared (34)};
    \draw[missed, line width=1.5pt] (0,-0.35) -- (0.5,-0.35);
    \node[anchor=west] at (0.6,-0.35) {missed optimal};
    \draw[wasted, line width=1.5pt] (0,-0.7) -- (0.5,-0.7);
    \node[anchor=west] at (0.6,-0.7) {wasted heuristic};
    \node[depot] at (0.25, -1.1) {};
    \node[anchor=west] at (0.6,-1.1) {depot};
  \end{tikzpicture}};

\end{tikzpicture}
\caption{GED--VRP decomposition on instance A-n36-k5 ($n{=}35$, $m{=}5$).
\textbf{(a)}~The complete graph $G_s$ has 630 edges with total weight
$W_{\text{total}} = 32{,}608$.
\textbf{(b)}~The optimal routing graph $G_t^*$ retains only 40~edges
($6.3\%$) with cost $C^* = 799$.
\textbf{(c)}~Comparing the heuristic solution ($C = 826$) against
the optimum: 34~edges are shared (blue), 6~optimal edges are
\emph{missed} (red dashed, total weight~189), and 6~non-optimal
edges are \emph{wasted} (orange dotted, total weight~216).
The cost gap $826 - 799 = 27$ equals wasted${}-{}$missed
$= 216 - 189 = 27$ exactly, as guaranteed by Theorem~1.
Edge weights on the missed and wasted edges show per-edge
cost attribution---a capability of the edge-additive GED objective.
Edge counts are reported here under the undirected convention
$|E|=n(n{+}1)/2$ used throughout \cref{sec:experiments}; under the
directed model of \cref{def:feasible-characterization} each undirected
stroke corresponds to two arcs.}
\label{fig:ged-decomp-instance}
\end{figure*}
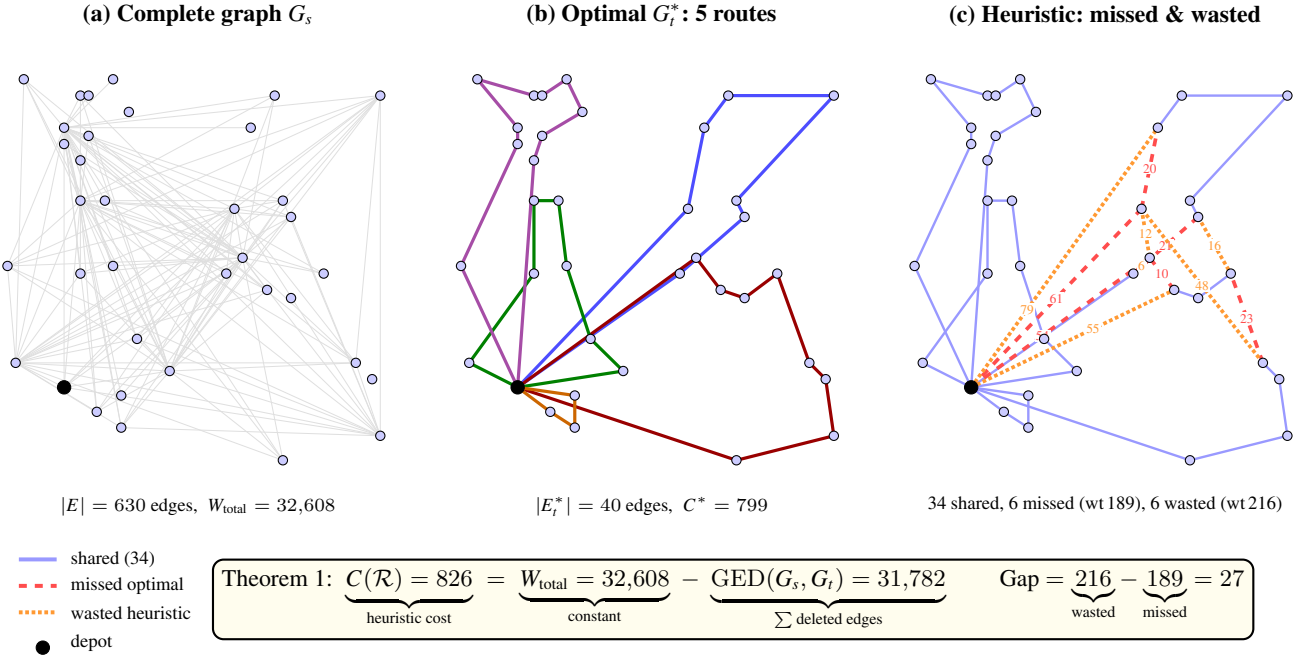
\begin{table}[t]
\centering
\caption{Gap decomposition: missed optimal edges vs.\ wasted
non-optimal edges across 90 instances (as \% of~$C^*$).}
\label{tab:decomp}
\begin{tabular}{lcc}
\toprule
\textbf{Component} & \textbf{Best solution} & \textbf{All solutions} \\
\midrule
Missed optimal edge weight  & 42.5\% & 56.2\% \\
Wasted non-optimal edge weight & 46.8\% & 71.7\% \\
Net gap (wasted $-$ missed) & 4.3\% & 15.5\% \\
\bottomrule
\end{tabular}
\end{table}

Even for the best solution, the heuristic \emph{replaces} 42.5\%
of the optimal cost with alternative edges costing 46.8\%---the
net gap of 4.3\% is a small residual of two large opposing flows.
Across all solutions, the flows are larger (56.2\%/71.7\%) but
the pattern is the same: heuristics substitute edges of
\emph{similar} weight, finding locally plausible but globally
suboptimal alternatives.

To determine whether missed edges reflect \emph{ordering} errors
(wrong sequence within a route) or \emph{assignment} errors
(customers placed in wrong routes), we classify each missed edge as
follows. For each missed edge $(i,j)\in E_t^*\setminus E_t$, let
$\rho^*(v)$ and $\rho(v)$ denote the route index containing customer
$v$ in the optimal and heuristic solution respectively. The edge is
an \emph{assignment error} if $\rho(i)\neq\rho^*(i)$ or
$\rho(j)\neq\rho^*(j)$ (at least one endpoint is in a different
route in the heuristic than in the optimum), and an \emph{ordering
error} otherwise (both endpoints are placed in their correct routes,
but the heuristic visits them in a non-optimal sequence). Across all
4{,}378 solution-instance pairs, \textbf{94.3\%} of missed edges are
assignment errors (median 97.6\%, exceeding 80\% in 93.2\% of
cases). This is a structural finding invariant across random
seeds: the dominant source of heuristic suboptimality is incorrect
route membership, not suboptimal intra-route ordering.

\medskip\noindent\textbf{Example (A-n36-k5).}
Across 50 heuristic solutions for this instance ($n{=}35$, $m{=}5$,
BKS${=}799$), the cost gap ranges from 16 to 220, with 4--28
missed optimal edges per solution.  The identity
gap $=$ wasted${}-{}$missed holds exactly in every case.
In all 50 solutions, the missed edges predominantly connect
customers to incorrect routes, confirming the assignment-error
diagnosis.  The specific missed edges vary across seeds, but the
structural conclusion is invariant: the GED framework characterises
the nature of the gap, not only its magnitude.
Moreover, \cref{fig:gap-decomp,fig:ged-decomp-instance} show the complete and optimal routing graphs as well as gap decomposition.

\begin{figure*}[h!]
\centering
\begin{tikzpicture}[
    >=Stealth,
    node distance=0.6cm and 0.8cm,
    inputbox/.style={
        draw, rounded corners=3pt, fill=blue!8,
        minimum height=0.85cm, minimum width=2.4cm,
        font=\small, align=center, thick
    },
    procbox/.style={
        draw, rounded corners=3pt, fill=orange!10,
        minimum height=0.85cm, minimum width=2.8cm,
        font=\small, align=center, thick
    },
    outputbox/.style={
        draw, rounded corners=3pt, fill=green!10,
        minimum height=0.85cm, minimum width=2.6cm,
        font=\small, align=center, thick
    },
    corebox/.style={
        draw, thick, fill=yellow!8, rounded corners=4pt,
        minimum height=1.1cm, minimum width=3.2cm,
        font=\small\bfseries, align=center
    },
    arr/.style={->, thick, >=Stealth},
    darr/.style={->, thick, >=Stealth, dashed, gray!70}
]

\node[inputbox] (inst) {VRP Instance\\[-1pt]$\mathcal{I}=(V,E,w,m,Q,D)$};
\node[inputbox, below=0.5cm of inst] (optsol) {Optimal Solution\\[-1pt]$\mathcal{R}^* \to G_t^*$};
\node[inputbox, below=0.5cm of optsol] (heur) {Heuristic Solutions\\[-1pt]$\hat{\mathcal{R}}^{(1)},\ldots,\hat{\mathcal{R}}^{(A)}$};

\node[corebox, right=1cm of optsol] (ged) {GED Formulation\\[-1pt]$C(\mathcal{R})=W_{\text{total}}-\text{GED}(G_s,G_t)$};

\node[procbox, right=1cm of ged, yshift=1.5cm] (sparse) {Sparsity Analysis\\[-1pt]$|E_t^*|/|E|$};
\node[procbox, right=1cm of ged, yshift=0.5cm] (overlap) {Edge Reachability\\[-1pt]CW frequency of $E_t^*$};
\node[procbox, right=1cm of ged, yshift=-0.5cm] (decomp) {Gap Decomposition\\[-1pt]missed vs.\ wasted};
\node[procbox, right=1cm of ged, yshift=-1.5cm] (approx) {Approx.\ Transfer\\[-1pt]$\alpha = \text{GED}_{\hat{R}}/\text{GED}^*$};

\node[outputbox, right=1cm of sparse] (o1) {$5.5\%$ edges used\\[-1pt]$97.7\%$ weight deleted};
\node[outputbox, right=1cm of overlap] (o2) {$8.7\%$ stable backbone\\[-1pt]$3.0\%$ not found in tested restarts};
\node[outputbox, right=1cm of decomp] (o3) {$42.5\%$ missed\\[-1pt]$46.8\%$ wasted};
\node[outputbox, right=1cm of approx] (o4) {$\alpha \geq 0.99$\\[-1pt]on all 90 instances};

\node[draw, thick, rounded corners=4pt, fill=red!6,
      minimum height=0.85cm, minimum width=3.0cm,
      font=\small, align=center,
      below=0.7cm of o4] (gnn)
      {GNN Edge Prediction\\[-1pt](future work)};

\draw[arr] (inst) -- (ged);
\draw[arr] (optsol) -- (ged);
\draw[arr] (heur) -- (ged);
\draw[arr] (ged) -- (sparse);
\draw[arr] (ged) -- (overlap);
\draw[arr] (ged) -- (decomp);
\draw[arr] (ged) -- (approx);
\draw[arr] (sparse) -- (o1);
\draw[arr] (overlap) -- (o2);
\draw[arr] (decomp) -- (o3);
\draw[arr] (approx) -- (o4);
\draw[darr] (o2.south) -- ++(0,-0.3) -| (gnn.north);
\draw[darr] (o1.south) -- ++(0,-0.55) -| (gnn.north);

\node[font=\scriptsize\bfseries, above=0.1cm of inst.north, text=blue!60!black] {Inputs};
\node[font=\scriptsize\bfseries, above=0.1cm of sparse.north, text=orange!70!black] {Structural Analyses};
\node[font=\scriptsize\bfseries, above=0.1cm of o1.north, text=green!50!black] {Findings};

\end{tikzpicture}
\caption{GED--VRP analytical framework.  A VRP instance, its optimal solution, and heuristic solutions are mapped through the edge-additive GED formulation (Theorem~1), enabling four structural analyses.  The sparsity and reachability findings define supervision signals for future GNN-based edge prediction.}
\label{fig:framework}
\end{figure*}
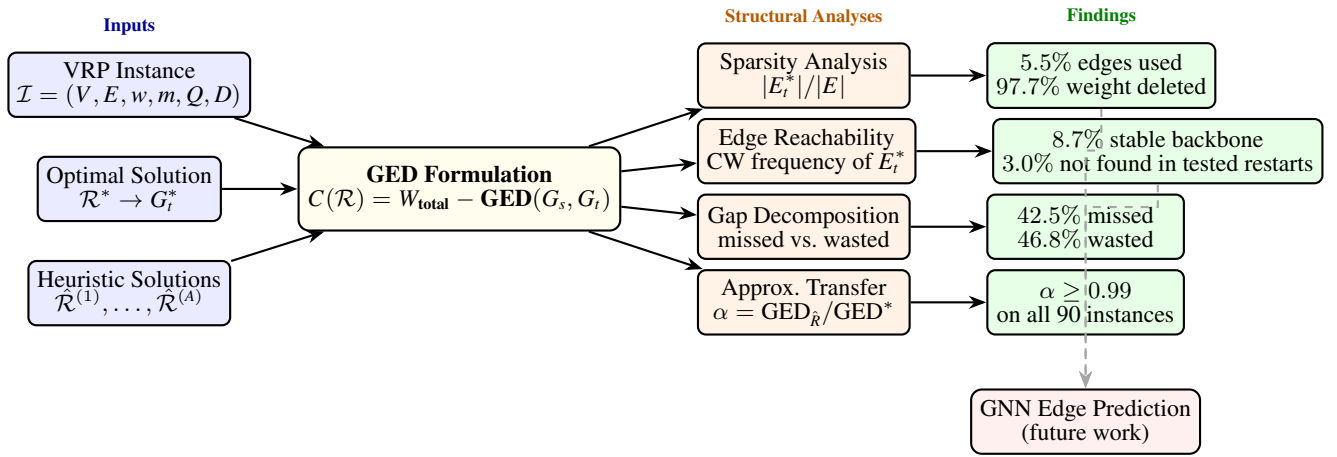

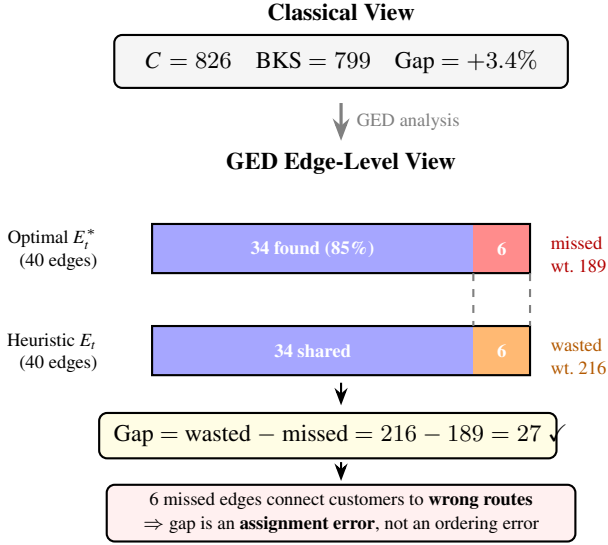
\begin{figure}[h!]
\centering
\begin{tikzpicture}[
    >=Stealth,
    font=\small
]

\node[font=\small\bfseries] at (2.5, 5.0) {Classical View};
\draw[thick, rounded corners=3pt, fill=gray!8]
    (-0.5, 4.0) rectangle (5.5, 4.7);
\node[align=center] at (2.5, 4.35)
    {$C = 826$\quad BKS $= 799$\quad Gap $= {+}3.4\%$};

\draw[->, very thick, gray] (2.5, 3.8) -- (2.5, 3.3)
    node[midway, right=2pt, font=\scriptsize, text=gray] {GED analysis};

\node[font=\small\bfseries] at (2.5, 3.0) {GED Edge-Level View};

\def\barw{5.0}
\def\barh{0.65}

\node[font=\scriptsize, anchor=east] at (-0.6, 2.0) {Optimal $E_t^*$};
\node[font=\scriptsize, anchor=east] at (-0.6, 1.7) {(40 edges)};

\pgfmathsetmacro{\sh}{34/40*\barw}    
\pgfmathsetmacro{\mi}{6/40*\barw}     

\fill[blue!35] (0, 1.55) rectangle (\sh, 1.55+\barh);
\fill[red!45]  (\sh, 1.55) rectangle (\sh+\mi, 1.55+\barh);
\draw[thick]   (0, 1.55) rectangle (\barw, 1.55+\barh);

\node[font=\scriptsize\bfseries, white] at (\sh/2, 1.87)
    {34 found (85\%)};
\node[font=\scriptsize\bfseries, white] at (\sh+\mi/2, 1.87)
    {6};

\node[font=\scriptsize, red!70!black, anchor=west] at (\barw+0.15, 1.95)
    {missed};
\node[font=\scriptsize, red!70!black, anchor=west] at (\barw+0.15, 1.65)
    {wt.\ 189};

\node[font=\scriptsize, anchor=east] at (-0.6, 0.65) {Heuristic $E_t$};
\node[font=\scriptsize, anchor=east] at (-0.6, 0.35) {(40 edges)};

\fill[blue!35]   (0, 0.2) rectangle (\sh, 0.2+\barh);
\fill[orange!55] (\sh, 0.2) rectangle (\sh+\mi, 0.2+\barh);
\draw[thick]     (0, 0.2) rectangle (\barw, 0.2+\barh);

\node[font=\scriptsize\bfseries, white] at (\sh/2, 0.52)
    {34 shared};
\node[font=\scriptsize\bfseries, white] at (\sh+\mi/2, 0.52)
    {6};

\node[font=\scriptsize, orange!70!black, anchor=west] at (\barw+0.15, 0.6)
    {wasted};
\node[font=\scriptsize, orange!70!black, anchor=west] at (\barw+0.15, 0.3)
    {wt.\ 216};

\draw[dashed, gray, thick] (\sh, 1.55) -- (\sh, 0.2+\barh);
\draw[dashed, gray, thick] (\sh+\mi, 1.55) -- (\sh+\mi, 0.2+\barh);

\draw[->, thick] (2.5, 0.1) -- (2.5, -0.2);
\draw[thick, rounded corners=3pt, fill=yellow!12]
    (-0.7, -0.9) rectangle (5.35, -0.3);
\node[align=center] at (2.5, -0.6)
    {Gap $=$ wasted $-$ missed $= 216 - 189 = 27$\;\checkmark};

\draw[->, thick] (2.5, -0.95) -- (2.5, -1.2);
\draw[thick, rounded corners=3pt, fill=red!6]
    (-0.6, -2.0) rectangle (5.6, -1.25);
\node[align=center, font=\scriptsize] at (2.5, -1.62)
    {6 missed edges connect customers to \textbf{wrong routes}\\[1pt]
     $\Rightarrow$ gap is an \textbf{assignment error}, not an ordering error};

\end{tikzpicture}
\caption{Gap decomposition for instance A-n36-k5 ($n{=}35$, $m{=}5$).
Classical analysis reports only the aggregate gap ($+3.4\%$).
The GED analysis decomposes this into 6~missed optimal edges
(weight~189) replaced by 6~non-optimal edges (weight~216),
with all missed edges connecting customers to incorrect routes.
The gap $27 = 216 - 189$ is verified exactly.}
\label{fig:gap-decomp}
\end{figure}
%
\subsection{Approximation Ratio}\label{sec:exp-approx}

For each heuristic solution, we compute the GED approximation ratio
$\alpha = \GED(G_s, \hat{G}_t) / \GED(G_s, G_t^*)$.
By \cref{thm:approximation}, this ratio bounds the VRP cost:
$C(\hat{R}) \leq (1{-}\alpha)\,W_{\mathrm{total}} + \alpha\, C(R^*)$.

\begin{table}[h!]
\centering
\caption{Approximation ratio $\alpha$ across 90 instances.}
\label{tab:approx}
\begin{tabular}{lcccc}
\toprule
\textbf{Statistic} & \textbf{Min} & \textbf{Median} & \textbf{Mean} & \textbf{Max} \\
\midrule
$\alpha$ & 0.9912 & 0.9995 & 0.9991 & 1.000 \\
\bottomrule
\end{tabular}
\end{table}

All 90 instances achieve $\alpha \geq 0.99$, confirming that
even a simple CW heuristic captures ${\geq}99\%$ of the optimal
GED.  The mean $\alpha = 0.9991$ demonstrates that \cref{thm:approximation}
provides a practically tight bound: the gap between heuristic and
optimal GED is less than $0.1\%$ on average, even though the
cost gap averages roughly 4\%.  This asymmetry arises because
$\GED^*$ dominates $W_{\mathrm{total}}$ (97.7\% on average), so
a small absolute GED difference maps to a larger relative cost
difference through the affine transformation of \cref{thm:main}.

\subsection{Discussion}

The experiments establish three empirical facts that require the
edge-additive structure of the GED formulation:

\begin{enumerate}
\item \textbf{Optimal solutions are extremely sparse in edge space.}
Only 5.5\% of edges are used, and 3.0\% of optimal edges are
not found by CW under tested restarts — establishing an empirical limitation of greedy construction.

\item \textbf{The cost gap has a precise edge-level decomposition.}
The 4.3\% average gap decomposes into 42.5\% missed optimal
edge weight replaced by 46.8\% wasted non-optimal weight.
This reveals that heuristic errors are primarily assignment
errors (wrong route), not ordering errors (within route) a
distinction invisible in the classical path-aggregated objective.

\item \textbf{The approximation bound is tight.}
$\alpha \geq 0.99$ universally, meaning GED approximation
algorithms can be directly applied to VRP with practically
useful cost guarantees via \cref{thm:approximation}.
\end{enumerate}

These analyses require reasoning about unused edges which edges
\emph{should} be in the solution but are not, and which edges are
in the solution but \emph{should not} be.  The GED formulation
provides this naturally through the complement structure of
\cref{rem:complement}: the binary partition of $E$ into $E_t$ and
$E \setminus E_t$ assigns every edge an unambiguous role (kept or
deleted) with an independent cost contribution.  Classical
formulations, which optimize over route sequences, do not provide
this edge-level decomposition without non-trivial post-processing.

The proposed GED-based formulation provides direct guidance for improving VRP solution methods. The observed sparsity of optimal solutions (about 5.5\% of edges) indicates that effective heuristics should focus on a restricted subset of promising edges rather than exploring the full graph.

Moreover, the decomposition of the optimality gap into missed and wasted edges offers a principled interpretation of local search behavior: improving solutions corresponds to increasing the inclusion of high-value edges while eliminating low-contribution ones.

The edge-additive structure of the formulation also naturally enables learning-based approaches. In particular, edge importance scores derived from GED can be used as supervision signals for models such as graph neural networks, guiding both constructive heuristics and refinement procedures.

Finally, the identification of optimal edges not found by the tested CW restarts highlights a practical limitation of classical heuristics such as Clarke–Wright, suggesting the need for hybrid approaches that combine heuristic construction with edge-level correction mechanisms.

Overall, these observations demonstrate that the proposed framework can be used to guide edge selection and evaluation within existing VRP heuristics. 
\section{GED Perspective on VRP Solution Methods}\label{sec:ged-perspective}
 
The equivalence $C(R) = W_{\mathrm{total}} - \GED(G_s, G_t)$
(\cref{thm:main}) provides more than an alternative objective function.
It reframes the VRP solution landscape in graph-theoretic terms
and opens directions that are difficult to pursue from the classical cost-minimisation
viewpoint.  This section develops three such directions.

\subsection{The Complement View: Optimising Unused Edges}
 
Classical VRP formulations optimise over \emph{used} edges
$E_t \subseteq E$.  The GED formulation inverts this: by
\cref{rem:complement}, the optimal routing graph $G_t^*$ has a
\emph{maximum-weight complement} $\bar{G}_t^* = (V, E \setminus E_t^*)$.
This is equivalent, but the complement view reveals structure that
the primal view hides.
 
For any feasible solution, define the \emph{complement deficiency}
of an edge $e \notin E_t$ as:
\begin{equation}\label{eq:deficiency}
\delta(e) = w(e) - \bar{w}_{\mathrm{avg}},
\end{equation}
where $\bar{w}_{\mathrm{avg}}$ is the average weight in
$E \setminus E_t$.  Edges with large negative deficiency are
cheap connections \emph{not} used by the solution.

One tempting reading on the limited benchmark instances, $58\%$ of the
100 cheapest complement edges connect customers in \emph{different}
routes. Under a uniform null model in which an unused edge spans two
random customers, the chance that those two customers lie in different
routes is approximately $(m-1)/m$, i.e. $80\%$ for $m{=}5$ and $90\%$
for $m{=}10$. The observed $58\%$ is therefore \emph{below} chance:
the cheapest complement edges are disproportionately \emph{within-route},
consistent with intra-route locality. The substantive use of the
complement view is not this cross-route fraction in isolation but rather
the per-edge attribution it enables in \cref{sec:exp-decomp}; the
within-route concentration of cheap complement edges is itself an
artifact of geographic clustering, not a misassignment signal.
 
\subsection{Edge Frequency and the GED Landscape}\label{sec:freq}
 
Running $A$ independent noisy CW restarts (\cref{eq:noisy}) produces
$A$ feasible solutions, each a point in the GED landscape.
The \emph{edge frequency} $f(e) = |\{R^{(a)} : e \in E_t^{(a)}\}|/A$
measures how consistently edge~$e$ appears across restarts.
 
\begin{proposition}[Frequency--GED Stability]\label{prop:freq}
An edge $e$ with $f(e) \approx 1$ has high $\Delta\GED$
(\cref{thm:merge}) under most merge orderings and is therefore
part of the stable GED backbone.  An edge with $f(e) \approx 0$
is a noise-dependent artifact.
\end{proposition}
 
Empirically, on $n{=}199$ instances, only $1\%$ of edges have
$f(e) \geq 0.8$, while $60\%$ appear in fewer than $10\%$ of
restarts.  This extreme sparsity of the consensus set explains
why multi-start methods outperform single-solution refinement:
the stable backbone is too small to guide local search, but large
enough to bias construction when used as a savings multiplier
$\tilde{s}(i,j) = s(i,j) \cdot (1 + \beta\, f(i,j))$,
where $\beta>0$ is a tunable boost factor (we use $\beta$ to avoid
collision with the GED approximation ratio $\alpha$ of
\cref{thm:approximation}). The empirical benefit of this boost
remains an open question for future work; we present the formula
as a \emph{proposed extension}, not as a validated technique.
 
\subsection{From GED to Learned Edge Prediction}\label{sec:gnn}
 
The GED formulation reduces VRP to an edge selection problem:
find the feasible subgraph $G_t \subseteq G_s$ maximising
$\sum_{e \in E \setminus E_t} w(e)$.  This formulation is
naturally suited to graph neural networks (GNNs), which operate
on edge-level features and produce edge-level predictions.
 
\subsubsection{Edge Classification Formulation}
For each edge $e \in E$, define the binary label
$y(e) = \mathbf{1}[e \in E_t^*]$ indicating membership in the
optimal routing graph.  A GNN $g_\theta : (V, E, \mathbf{X}) \to [0,1]^{|E|}$
can be trained to predict $\hat{y}(e) = g_\theta(e)$ using the
GED-derived loss:
\begin{equation}\label{eq:gnn-loss}
\mathcal{L}(\theta) = \sum_{e \in E} w(e) \cdot
\ell\!\left(\hat{y}(e),\, y(e)\right),
\end{equation}
where $\ell$ is the binary cross-entropy and the weighting by $w(e)$
ensures that mispredicting high-weight edges incurs proportionally
higher loss directly reflecting the GED objective.
 
\subsubsection{GNN-Warm-Started Construction}
Given trained predictions $\hat{y}(e)$, the CW merge loop can be
modified to incorporate learned edge scores:
\begin{equation}\label{eq:gnn-savings}
s_{\mathrm{GNN}}(i,j) = s(i,j) \cdot \bigl(1 + \beta\,\hat{y}(i,j)\bigr),
\end{equation}
boosting the savings of edges predicted to be in $E_t^*$.
This replaces the statistical edge frequency signal
(\cref{sec:freq}) with a \emph{learned} signal that generalises
across instances.  Unlike frequency-based boosting, which requires
$A$ restarts on the \emph{current} instance, GNN predictions are
computed in a single forward pass and transfer to unseen instances
of similar structure.
 
\subsubsection{Feasibility via Degree Constraints}
The feasibility conditions of \cref{def:feasible-characterization},$\deg^+(0){=}\deg^-(0){=}m$
and $\deg^+(i){=}\deg^-(i){=}1$ for all customers, are linear constraints on edge
predictions.  They can be enforced either as a differentiable
penalty during training or as a post-hoc projection step that rounds
$\hat{y}(e)$ to a feasible binary assignment.

The GNN-based edge prediction framework described in this section is a research direction enabled by the GED formulation. No GNN experiments are included in this paper; establishing whether the GED-weighted loss and sparsity signal yield practical improvements over existing edge-classification labels (such as those used in~\cite{kool2019attention,joshi2020learning}) remains an open empirical question for future work.

\subsection{Approximation-Aware Algorithm Design}
 
\Cref{thm:approximation} establishes that a GED $\alpha$-approximation
yields a VRP solution with cost
$C(\hat{R}) \leq (1{-}\alpha)W_{\mathrm{total}} + \alpha\, C(R^*)$.
This bound has a practical consequence: any algorithm that
\emph{guarantees} $\GED(G_s, \hat{G}_t) \geq \alpha \cdot \GED^*$
automatically provides a VRP cost guarantee.
 
The extensive GED approximation literature, like bipartite matching
\cite{fankhauser2011speeding}, beam search on edit paths \cite{neuhaus2006fast},
and neural GED estimators \cite{bai2020simgnn,li2019graph}, thus becomes
directly applicable to VRP.  In particular, bipartite GED
approximations run in $O(n^3)$ and provide known approximation
ratios, which \cref{thm:approximation} translates into VRP cost bounds.
Exploring these connections is a promising direction for future work.

\section{Related Work}\label{sec:related}

\subsection{Vehicle Routing Problem}

The VRP was introduced by Dantzig and Ramser~\cite{dantzig1959truck} as the ``Truck Dispatching
Problem'' and has been studied extensively for over six decades. Exact methods include
branch-and-cut~\cite{gelinas1995branching} and branch-and-price~\cite{desrochers1988vehicle}, which
remain effective for small to medium instances. Among construction heuristics, the Clarke--Wright
savings algorithm~\cite{clarke1964scheduling} is the most relevant to this work: it builds routes by
iteratively merging tours in descending savings order, and Theorem~2 establishes that each such merge
is exactly a per-merge GED increment. We use it as the baseline construction heuristic throughout
Section~VII.

For improvement, metaheuristics such as tabu search~\cite{gendreau1994tabu}, adaptive large
neighborhood search (ALNS)~\cite{ropke2006adaptive}, and genetic
algorithms~\cite{potvin1996genetic} achieve high-quality solutions on large instances. The
Lin--Kernighan--Helsgaun heuristic (LKH-3)~\cite{helsgaun2000effective,helsgaun2017extension} is the
representative high-performance solver for CVRP and related variants, cited here as a reference for
solution quality. Learning-based approaches---pointer networks~\cite{vinyals2015pointer}, attention
models~\cite{kool2019attention}, and graph neural networks
(GNNs)~\cite{joshi2020learning,nazari2018reinforcement}---have been applied to route construction or
improvement. These are directly relevant to this work: the GED reformulation provides a natural
per-edge supervision signal (Equation~(23)) and a well-posed loss function for GNN-based edge
prediction, bridging the gap between the GNN literature and combinatorial VRP. Evolutionary and
swarm-based methods have also been applied, including hybrid genetic algorithms for reconfigurable
networks~\cite{akopov2025} and particle swarm optimization for clustered
routing~\cite{islam2021}.

\subsection{Classical VRP Formulations}

The CVRP formulation originates with Dantzig and Ramser~\cite{dantzig1959truck}. Modern taxonomies of
mathematical programming structures for CVRP are reviewed in Toth and Vigo~\cite{toth2002models},
which categorises edge-flow and vehicle-flow formulations by index aggregation, commodity management,
and bounding properties. These classical formulations are the backdrop against which the GED
reformulation is contrasted: all of them optimise over route sequences, making per-edge decomposition
of the objective unavailable without non-trivial post-processing.

\smallskip
\noindent\textbf{Two-index vehicle flow.}
The canonical two-index edge-flow formulation, advanced by Laporte and
Nobert~\cite{laporte1987exact}, uses a binary variable $x_{ij}$ indicating whether edge $(i,j)$ is
traversed by any vehicle in the fleet, avoiding vehicle-index symmetries. Subtour elimination and
capacity constraints are enforced via capacity cut constraints (CCCs), which yield tight LP relaxations
and underpin branch-and-cut algorithms despite growing exponentially in number.

\smallskip
\noindent\textbf{Three-index vehicle flow.}
Multi-index formulations append an explicit vehicle index $k$ to decision variables, yielding
$x_{ijk} \in \{0,1\}$. This supports heterogeneous fleets but introduces permutation symmetry
in the solution space, causing severe scaling limitations in branch-and-bound engines.

\smallskip
\noindent\textbf{MTZ constraints.}
Miller, Tucker, and Zemlin~\cite{miller1960integer} proposed polynomial-sized subtour elimination via
continuous auxiliary variables $u_i$ representing cumulative load:
\begin{equation*}
  u_i - u_j + C x_{ij} \leq C - q_j \quad \forall (i,j) \in A,\; i,j \neq 0.
\end{equation*}
While this reduces the constraint count to $O(|V|^2)$, the big-$M$ structure yields weak LP
relaxations.

\smallskip
\noindent\textbf{Commodity flow.}
Gavish and Graves~\cite{gavish1981scheduling} pioneered commodity flow formulations that treat residual
vehicle capacity as a continuous network flow, producing significantly tighter LP bounds than MTZ
without requiring exponential constraint separation.

\smallskip
These four formulations establish the standard landscape of VRP mathematical programming. The GED
reformulation presented in Section~III differs from all of them: its objective is edge-additive,
enabling per-edge quality attribution and direct application of GED approximation bounds.

\subsection{Graph Edit Distance}

GED is a classic method for error-tolerant graph matching~\cite{sanfeliu1983distance}, defined as
the minimum-cost sequence of node/edge insertions, deletions, and substitutions that transforms one
graph into another. Exact GED computation is NP-hard and is typically solved via A* search or integer
programming~\cite{neuhaus2006fast,zeng2009comparing}. Approximation methods based on bipartite
matching~\cite{fankhauser2011speeding,neuhaus2006fast} trade optimality for polynomial runtime.
Efficient exact and approximate approaches exploiting parallel and GPU architectures have also been
developed~\cite{dabah2022parallel,dabah2021approximate,dabah2026gpu}.

Learning-based GED approximations using GNNs such as SimGNN~\cite{bai2020simgnn} and subsequent
architectures~\cite{li2019graph,piao2022learning}, learn graph similarity end-to-end and are directly
relevant to Section~VIII-C, where we identify the GED formulation as a natural basis for GNN-based
edge prediction in VRP.

Unlike prior GED research, which treats edit distance as a similarity measure to be minimized or
approximated between two given graphs, this work reformulates a combinatorial optimization problem
(VRP) as a GED \emph{maximization} problem relative to a complete source graph. This shift enables
the application of known GED approximation algorithms~\cite{fankhauser2011speeding,neuhaus2006fast} and k-best search~\cite{dabah2021approximate} to VRP with practical structural analyses metrics.

\section{Conclusion}\label{sec:conclusion}

We have shown that the Vehicle Routing Problem can be reformulated as a Graph Edit Distance maximization problem. Under a simple edge‑deletion cost model, minimizing route cost is equivalent to maximizing the total weight of edges deleted from the complete instance graph. This reformulation yields three theoretical results: a formal equivalence (\cref{thm:main}), a merge‑decomposition theorem linking Clarke‑Wright savings to per‑merge GED increments (\cref{thm:merge}), and an approximation‑transfer theorem that translates GED approximation ratios into VRP cost bounds (\cref{thm:approximation}).

The power of this reformulation is demonstrated through structural analyses that are difficult to obtain from the classical cost‑minimization perspective. Validated on 90 benchmark instances with proven optimal solutions, we establish three empirical facts:  
(i) optimal routing graphs use only 5.5\% of available edges;  
(ii) 3.0\% of optimal edges were not found by Clarke‑Wright under the tested restarts, revealing an empirical limitation of greedy methods; and  
(iii) across all heuristic solutions, 94.3\% of missed optimal edges are assignment errors.

These edge‑level insights define potential targets for future work: developing GNN models that predict the $\sim$5.5\% of edges belonging to $E_t^*$ using the GED‑weighted loss of \cref{eq:gnn-loss}; exploiting the GED approximation literature (bipartite matching, beam search on edit paths) to obtain VRP solutions with formal cost guarantees via \cref{thm:approximation}; and extending the formulation to VRPTW and stochastic VRP variants through the constraint‑encoding framework.

Finally, the proposed GED formulation does not introduce additional computational overhead and could be incorporated into existing VRP heuristics as an evaluation or guidance layer without modifying their core structure; demonstrating an end-to-end integration in which GED-derived signals measurably improve a solver is left for future work.
\section*{Data Availability}
Code will be released upon acceptance.
\bibliographystyle{IEEEtran}
\bibliography{ref}
\appendix

\section{VRP Constraints Encoding in GED}\label{sec:constraint-theory}

\begin{definition}[Constrained GED]\label{def:constrained-ged}
Given a set of feasibility constraints \(\mathcal{C}\), the constrained GED is:
\begin{equation}
\GED_{\mathcal{C}}(G_s, G_t) = \sum_{e \in E \setminus E_t} w(e) + \sum_{e \in E_t} \phi_{\mathcal{C}}(e)
\end{equation}
where \(\phi_{\mathcal{C}}: E \to \mathbb{R}^+ \cup \{\infty\}\) is a penalty function that is \(0\) if \(e\) belongs to some feasible routing graph satisfying \(\mathcal{C}\) and \(\infty\) otherwise, with intermediate values for partial violations.
\end{definition}

\subsection{Capacity Constraints}

\begin{definition}[Capacity Violation Indicator]\label{def:capacity-violation}
For an edge \(e = (i,j)\), define the capacity violation indicator:
\begin{equation}
\mathbb{1}_{\text{cap}}(e) = 
\begin{cases}
1 & \text{if including } e \text{ in a route would make it} \\
& \text{  infeasible given the remaining vertices} \\
0 & \text{otherwise}
\end{cases}
\end{equation}
\end{definition}

\begin{theorem}[Capacity Encoding]\label{thm:capacity}
There exists a penalty function \(\lambda: E \to \mathbb{R}^+\) such that:
\[
\min_{\substack{G_t \text{ degree-feasible} \\ \text{capacity constraints}}} \GED(G_s, G_t) = \]
\[
\min_{G_t \text{ degree-feasible}} \left( \GED(G_s, G_t) + \sum_{e \in E_t} \lambda(e) \right)
\]
where \(\lambda(e)\) can be chosen as \(\lambda(e) = M \cdot \mathbb{1}_{\text{cap}}(e)\) for sufficiently large \(M > \wt\).
\end{theorem}

\begin{proof}
For any degree-feasible graph \(G_t\) that violates capacity constraints, there exists at least one edge \(e \in E_t\) with \(\mathbb{1}_{\text{cap}}(e) = 1\). The penalty term adds at least \(M\) to the objective. Since \(\GED(G_s, G_t) \leq \wt\) and \(M > \wt\), any capacity-violating graph has strictly larger penalized GED than any feasible graph. Therefore, the minimizer of the penalized objective must satisfy capacity constraints.
\end{proof}

\subsection{Time Windows}

\begin{definition}[Time Window Violation]\label{def:time-violation}
For an edge \(e = (i,j)\) with time windows \([a_i, b_i]\) and \([a_j, b_j]\), and travel time \(t(i,j)\), define:
\begin{equation}
\tau(e) = \max(0, a_j - (b_i + t(i,j)))
\end{equation}
as the minimum waiting time required to make the edge feasible.
\end{definition}

\begin{proposition}[Time Window Penalty]\label{prop:time}
The time window constraint can be encoded by modifying deletion costs:
\begin{equation}
c_{\text{del}}'(e) = w(e) + \mu \cdot \tau(e)
\end{equation}
where \(\mu > 0\) is a scaling parameter. As \(\mu \to \infty\), the penalized formulation approaches the hard constraint case.
\end{proposition}

\subsection{Heterogeneous Fleet}

\begin{definition}[Heterogeneous GED]\label{def:heterogeneous}
For a fleet with vehicle-specific cost matrices \(w_k\) for \(k = 1,\ldots,m\), define:
\begin{equation}
\GED_{\text{hetero}}(G_s, G_t) = \sum_{k=1}^m \sum_{e \in E \setminus E_t^k} w_k(e)
\end{equation}
where \(E_t^k\) are edges assigned to vehicle \(k\), with \(\bigcup_{k=1}^m E_t^k = E_t\) and \(E_t^k \cap E_t^j = \emptyset\) for \(k \neq j\).
\end{definition}
\section*{Acknowledgment}
The author used AI-based writing assistance to enhance the clarity, grammar, and presentation of this manuscript.
\begin{IEEEbiography}
[{\includegraphics[width=1in,height=1.25in,clip,trim=500 0 500 0,keepaspectratio]{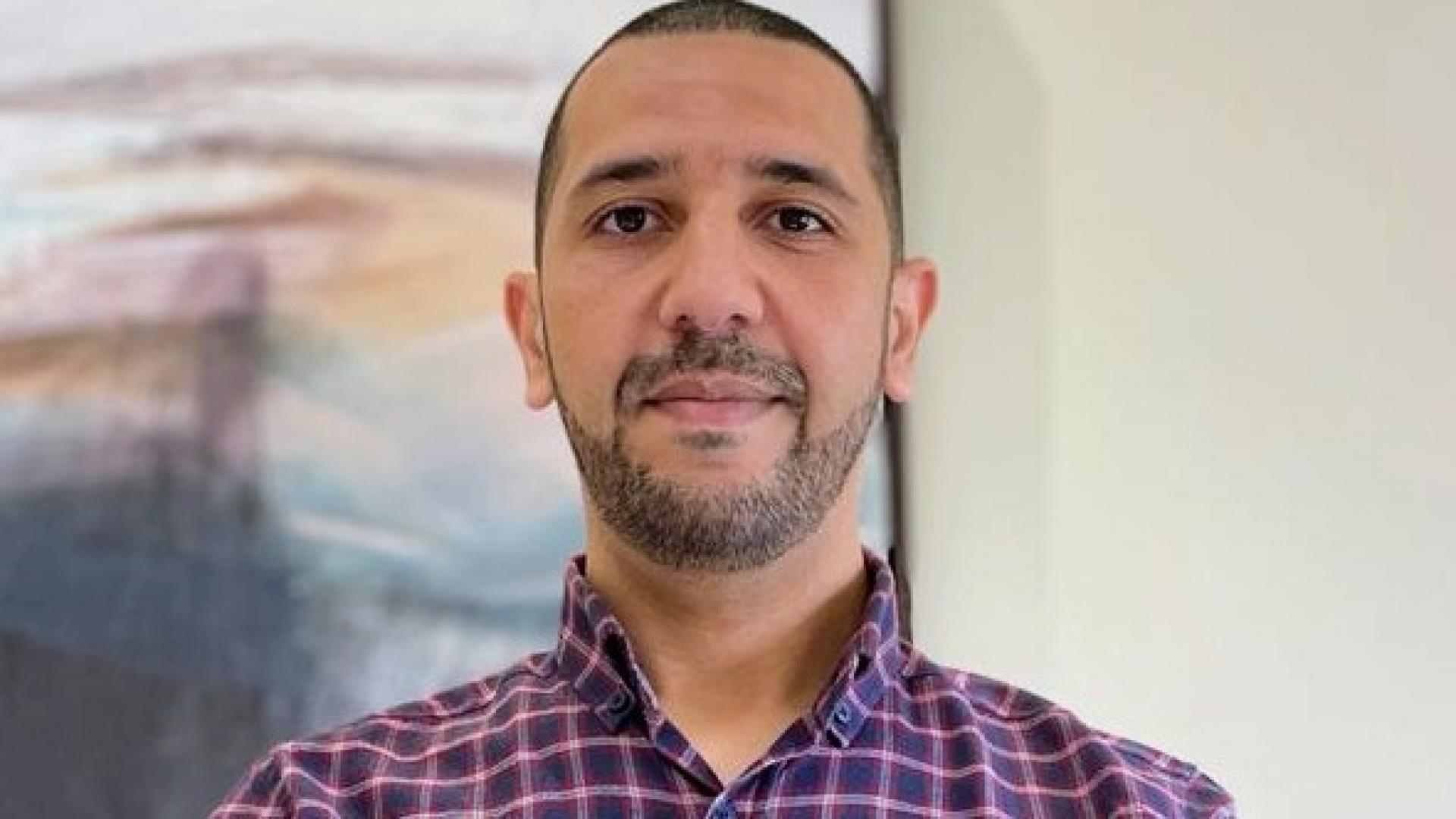}}]
{Dr. Adel Dabah}received the B.S. degree in Computer Science from USTHB University, Algeria, in 2010, and the M.S. and Ph.D. degrees in Computer Science from USTHB University, 2013 and 2018, respectively. He is currently a Research Scientist in the Accelerating Devices Lab at Julich Supercomputing Centre (JSC), Germany.  His research interests include High-Performance Computing, Scheduling Problems, Graph Edit Distance,  Discrete algorithms, and software development for scientific applications.
\end{IEEEbiography}
\EOD
\end{document}